\documentclass[pra,twocolumn,superscriptaddress]{revtex4-1}%
\usepackage{amsfonts}
\usepackage{amsmath,amsthm}
\usepackage{amssymb}
\usepackage[colorlinks=true,linkcolor=blue,citecolor=blue,plainpages=false,pdfpagelabels]{hyperref}

\usepackage{graphicx}%
\usepackage[caption=false]{subfig}

\usepackage{overpic}

\providecommand{\U}[1]{\protect\rule{.1in}{.1in}}
\newtheorem{theorem}{Theorem}

\newtheorem{lemma}[theorem]{Lemma}

\newtheorem{proposition}[theorem]{Proposition}

\begin{document}
\preprint{ }

\title{Gaussian hypothesis testing and quantum illumination}
%\title{Asymmetric hypothesis testing of Gaussian states and quantum illumination}
\author{Mark M. Wilde}
\affiliation{Hearne Institute for Theoretical Physics, Department of Physics and Astronomy,
Center for Computation and Technology, Louisiana State University, Baton
Rouge, Louisiana 70803, USA}

\author{Marco Tomamichel}
\affiliation{Centre for Quantum Software and Information and School of Software, University of Technology Sydney, Broadway, NSW 2007, Australia}

\author{Seth Lloyd}
\affiliation{Research Laboratory of Electronics and the Department of Mechanical
Engineering, Massachusetts Institute of Technology, Cambridge, Massachusetts
02139, USA}

\author{Mario Berta}
\affiliation{Institute for Quantum Information and Matter, California Institute of
Technology, Pasadena, California 91125, USA}

\keywords{}
\pacs{PACS number}

\begin{abstract}
Quantum hypothesis testing is one of the most basic tasks in quantum
information theory and has fundamental links with quantum communication
and estimation theory. In this paper, we establish a formula that
characterizes the decay rate of the minimal Type-II error probability in a
quantum hypothesis test of two Gaussian states given a fixed constraint on the
Type-I error probability. This formula is a direct function of the mean
vectors and covariance matrices of the quantum Gaussian states in question. We
give an application to quantum illumination, which is the task
of determining whether there is a low-reflectivity object embedded in a target
region with a bright thermal-noise bath. For the asymmetric-error setting, we find that a quantum illumination
transmitter can achieve an error probability exponent stronger than a
coherent-state transmitter of the same mean photon number, and furthermore,
that it requires far fewer trials to do so. This occurs when the background thermal noise is either low or bright, which means that a quantum advantage is even easier to witness than in the symmetric-error setting because it occurs for a larger range of parameters. Going forward from here, we expect our formula to have applications in settings well beyond those considered in
this paper, especially to quantum communication tasks involving quantum
Gaussian channels.
\end{abstract}

\volumeyear{ }
\volumenumber{ }
\issuenumber{ }
\eid{ }
\date{\today}
\startpage{1}
\endpage{102}
\maketitle

\textit{Introduction}---Hypothesis testing is critical for the scientific
method \cite{LR08}, underlying our ability to distinguish various
models of reality and draw conclusions accordingly. 
It also has fundamental links with
both communication~\cite{B74} and estimation theory~\cite{B60}. By
increasing the number of independent samples observed in a given experimental
setup, one can reduce the probability of making an incorrect inference, thus
increasing the confidence in conclusions drawn from the experiment.

In the most basic setting of binary hypothesis testing the goal is to distinguish
two hypotheses (null and alternative). There are two ways that one can
err: a Type-I error (``false alarm'')
occurs when rejecting the null hypothesis when it is in fact
true, and analogously a Type-II error (``false
negative'') occurs when incorrectly rejecting the alternative hypothesis. If it is
possible to obtain many independent samples, one can study how error
probabilities decay as a function of the number of samples for an optimal sequence of tests.
Most prominently, the Chernoff bound~\cite{chernoff1952} tells us that both error probabilities decay exponentially fast (in the number of samples) for an appropriately chosen sequence of tests. Beyond this, it is often desirable to treat the two types of errors asymmetrically. For example, the experimenter
may only require a fixed bound on the ``false alarm'' probability and
then seek to minimize the ``false negative'' probability subject to this
constraint. The well known result here is the Chernoff--Stein lemma (sometimes
called Stein's lemma) \cite{chernoff1952}, which establishes how fast the ``false negative'' probability decays in this setting.

Since the rise of quantum information science, researchers have generalized these notions to the fully quantum setup,
which is arguably more fundamental than the classical settings discussed
above. Here the basic setting involves determining whether $M\geq1$
quantum systems are described by the density operator $\rho^{\otimes M}$ or
another density operator $\sigma^{\otimes M}$, and the experimenter is allowed
to perform a collective quantum measurement on all $M$ systems in order to
guess which is the case. The fundamental results are the quantum Chernoff
bound \cite{PhysRevLett.98.160501,ANSV08}, which states that the quantum
Chernoff information is the optimal decay rate when minimizing both error
probabilities simultaneously, and the quantum Stein's lemma \cite{HP91,ON00},
which states that the quantum relative entropy between $\rho$ and
$\sigma$ is the optimal decay rate for the Type-II error probability given a
fixed (independent of~$M$) constraint on the Type-I error probability. In more recent years, we have seen strong refinements of quantum
Stein's lemma \cite{li12,TH12,CMMAB08,JOPS12,DPR15} that characterize the decay in higher orders of $M$ and are crucial for a finite-size analysis.

One of the major applications of the results of quantum hypothesis testing is
quantum illumination \cite{L08}. In the setting of quantum illumination, a
source emits photons entangled in signal and idler beams, and the signal beam
is subsequently subjected to a modulation, loss, and environmental noise. A
quantum receiver then makes a collective measurement on both the returned
signal and idler beams in order to determine which modulation was applied. The
typical task considered in previous work is to determine whether a target
region containing a bright thermal-noise bath has a low-reflectivity object
embedded \cite{L08,TEGGLMPS08}. Alternatively, one could also use the quantum
illumination setup as a secure communication system, as proposed in
\cite{S09}. After the original proposal of quantum illumination \cite{L08}, a
full Gaussian state treatment appeared
\cite{TEGGLMPS08} and strengthened the predictions of \cite{L08}. The upshot
is that quantum illumination can offer a significant performance advantage
over a classical coherent-state transmitter of the same average photon number,
when considering the sensing application mentioned above. To date, several
experiments have been conducted that demonstrate the advantage quantum
illumination offers \cite{LRDOBG13,ZTZWS13,ZMWS15,BGWVSP15}.

Hitherto quantum illumination has mostly been considered in the symmetric-error setting \cite{TEGGLMPS08,BGWVSP15},
and as such, one of the main technical tools
employed in the analysis of quantum illumination is the quantum Chernoff
bound. However, there are many scenarios where one is interested in the performance of quantum
illumination in the asymmetric-error setting. Indeed, one might be willing to accept a particular Type-I\ error (``false alarm'') probability (the error being to declare a target present when in fact it is not), and then minimize the Type-II\ error (``false negative'') probability subject to this constraint.

In this paper, we determine the second-order refinement of quantum Stein's Lemma in Gaussian quantum hypothesis testing. As our main result we derive an analytical formula that expresses the second-order behavior for any two Gaussian states as a function of their vector means and covariance matrices. Our result has applications to quantum illumination in the asymmetric-error setting, where we find that there are significant gains over a classical coherent-state emitter. Notably, we find that the quantum advantage is even easier to witness than in the symmetric-error setting because it occurs for a larger range of parameters.

We expect our formula to have applications well
beyond the setting considered here, to various tasks in quantum communication
theory. In fact, it is the basis for the
strongest known upper bounds on quantum key distribution protocols conducted
over quantum Gaussian channels \cite{WTB16,KW17}. In light of this, we expect our result to be useful in establishing sharp refinements of various capacities
of quantum Gaussian communication channels, when combined with generalizations
of the methods from \cite{TT13,WRG16,DTW14,BDL15,TBR15}.

To elaborate on our main result, if the task is to
distinguish $\rho^{\otimes M}$ from $\sigma^{\otimes M}$ and the Type-I error
cannot exceed $\varepsilon\in\left(  0,1\right)$, then the optimal Type-II
error probability $\beta$ takes the exponential form
\begin{equation}
\exp\left[  -\left(  Ma+\sqrt{Mb}\Phi^{-1}(\varepsilon)+O(\ln M)\right)
\right]  .\label{eq:Gaussian-approx}
\end{equation}
The optimal constant $a\geq0$ was identified in \cite{HP91,ON00} to be the quantum relative entropy \cite{U62},
defined as
$
a=D(\rho\Vert\sigma)\equiv\left\langle \ln\rho-\ln\sigma\right\rangle _{\rho}
$
for faithful $\sigma$ where we used the convention $\left\langle \cdot \right\rangle _{\rho}
\equiv\operatorname{Tr}\{\rho\, \cdot \}$. The optimal constant $b\geq0$ was
identified in \cite{li12,TH12,DPR15} to be the quantum relative entropy
variance, defined in terms of the variance of the operator $\ln\rho-\ln\sigma$
\begin{equation}
b=V(\rho\Vert\sigma)\equiv\langle\left[  \ln\rho-\ln\sigma-D(\rho\Vert
\sigma)\right]  ^{2}\rangle_{\rho}\label{eq:relative-entropy-variance} \,.
\end{equation}
In the above, we have also used the cumulative distribution function for a
standard normal random variable:
$
\Phi(y)   \equiv\frac{1}{\sqrt{2\pi}}\int_{-\infty}^{y}dx\,\exp\left(
-x^{2}/2\right)  $.
The derivation of~\cite{DPR15} also applies to particular states on separable Hilbert spaces~\cite{datta_private}, of which Gaussian states are examples.

An explicit formula for the quantum relative entropy between any two Gaussian
states, as a function of their mean vectors and covariance matrices, was given
in \cite{PhysRevA.71.062320}\ and refined in \cite{PLOB15}. Here we derive an
explicit formula for the quantum relative entropy
variance of two Gaussian states, given as a function of their mean vectors and
covariance matrices. The formula allows for a deeper understanding of quantum hypothesis
testing of Gaussian states.
We state our result after a brief recollection of
the Gaussian state formalism (see \cite{Arvind1995,adesso14}\ for detailed reviews),
and provide a detailed proof in the appendix. Finally, we apply our formula in the context of quantum illumination, giving a
characterization of its performance in the asymmetric-error setting.

\textit{Related work}---The authors of \cite{SB14} considered asymmetric
hypothesis testing of quantum Gaussian states, deriving a formula for
the quantum Hoeffding bound \cite{PhysRevA.76.062301,N06,ANSV08,0305-4470-35-50-307}\ in the
context of Gaussian state discrimination. However, the setting of the quantum
Hoeffding bound is conceptually different from what we consider here.

\textit{Gaussian state formalism---}We begin by reviewing some background on
Gaussian states and then review a formula for quantum relative entropy from
\cite{PhysRevA.71.062320,PLOB15}\ (see \cite{Arvind1995,PLOB15} for more details on the conventions used). Our development applies to $n$-mode Gaussian states, where $n$ is some fixed
positive integer. Let $\hat{x}_{j}$ denote each quadrature operator ($2n$ of
them for an $n$-mode state), and let $\hat{x}\equiv\left[  \hat{q}_{1}
,\ldots,\hat{q}_{n},\hat{p}_{1},\ldots,\hat{p}_{n}\right]  \equiv\left[
\hat{x}_{1},\ldots,\hat{x}_{2n}\right]  $ denote the vector of quadrature
operators, so that the first~$n$ entries correspond to position-quadrature
operators and the last~$n$ to momentum-quadrature operators. The quadrature
operators satisfy the commutation relations:
\begin{equation}\label{eq:symplectic-form}
\left[  \hat{x}_{j},\hat{x}_{k}\right]  =i\Omega_{j,k},
\end{equation}
where $\Omega=
\begin{bmatrix}
0 & 1\\
-1 & 0
\end{bmatrix}
\otimes I_{n}$ and $I_{n}$ is the $n\times n$ identity matrix. We also take the annihilation
operator $\hat{a}=\left(  \hat{q}+i\hat{p}\right)  /\sqrt{2}$. Let $\rho$ be a
Gaussian state, with the mean-vector entries $\left\langle \hat{x}_{j}\right\rangle _{\rho}=\mu_{j}^{\rho}$, and let $\mu^{\rho}$ denote the mean vector. The entries of the Wigner
function covariance matrix $V^{\rho}$\ of $\rho$ are given by
\begin{equation}
V_{j,k}^{\rho}\equiv\frac{1}{2}\left\langle \left\{  \hat{x}_{j}-\mu_{j}
^{\rho},\hat{x}_{k}-\mu_{k}^{\rho}\right\}  \right\rangle _{\rho}.
\label{eq:covariance-matrices}
\end{equation}
A $2n\times2n$ matrix $S$ is symplectic if it preserves the symplectic form:
$S\Omega S^{T}=\Omega$. According to Williamson's theorem \cite{W36}, there is
a diagonalization of the covariance matrix $V^{\rho}$ of the form,
$
V^{\rho}=S^{\rho}\left(  D^{\rho}\oplus D^{\rho}\right)  \left(  S^{\rho
}\right)^{T},
$
where $S^{\rho}$ is a symplectic matrix and $D^{\rho}\equiv\operatorname{diag}
(\nu_{1},\ldots,\nu_{n})$ is a diagonal matrix of symplectic eigenvalues such
that $\nu_{i}\geq1/2$ for all $i\in\left\{  1,\ldots,n\right\}  $. We can
write the density operator $\rho$ in the exponential form
\cite{PhysRevA.71.062320,K06,H10,H13book},
\begin{align}
&\rho=Z_{\rho}^{-1/2}\exp\left[  -\frac{1}{2}(\hat{x}-\mu^{\rho})^{T}G_{\rho
}(\hat{x}-\mu^{\rho})\right], \label{eq:exp-form}\\
&\mathrm{with}\quad 
\quad G_{\rho}  \equiv-2\Omega S^{\rho}\left[  \operatorname{arcoth}(2D^{\rho
})\right]  ^{\oplus2}\left(  S^{\rho}\right)  ^{T}\Omega, \label{eq:G_rho}
\end{align}
and $Z_{\rho}   \equiv\det(V^{\rho}+i\Omega/2)$, where $\operatorname{arcoth}(x)\equiv\frac{1}{2}\ln\!\left(  \frac{x+1}
{x-1}\right)$ with domain $\left(  -\infty,-1\right)  \cup\left(
1,+\infty\right)$.
%(here we exclusively consider the subset $\left(1,+\infty\right)$ of the domain).
Note that we can also write
$
G_{\rho}=2i\Omega\operatorname{arcoth}(2iV^{\rho}\Omega),
%\label{eq:more-compact-G-rho}
$
so that $G_{\rho}$ is represented directly in terms of the
covariance matrix $V^{\rho}$ (see Appendix~\ref{sec:matrix-diagonalizable} on how to compute the symplectic
decomposition of $V^{\rho}$). By inspection, the $G$ and $V$ matrices are
symmetric, which is critical in our analysis below. As a result,
$\operatorname{Tr}\{G\Omega\}=\operatorname{Tr}\{V\Omega\}=0$ because $G$ and
$V$ are symmetric while $\Omega$ is antisymmetric. In what follows, we adopt
the same notation for quantities associated with a density operator $\sigma$,
such as $\mu^{\sigma}$, $V^{\sigma}$, $S^{\sigma}$, $D^{\sigma}$, $Z_{\sigma}
$, and $G_{\sigma}$.

\textit{Relative entropy for Gaussian states}---We first revisit the relative
entropy calculation from \cite{PhysRevA.71.062320}, but following the
particular aspects of \cite{PLOB15}. Suppose for simplicity that $\rho$ and
$\sigma$ are zero-mean Gaussian states. By employing the exponential form in
\eqref{eq:exp-form}, we see that
\begin{equation}
\left\langle \ln\rho-\ln\sigma\right\rangle _{\rho}=\frac{1}{2}\left[  \ln
Z_{\sigma}-\ln Z_{\rho}-\left\langle \hat{x}^{T}\Gamma\hat{x}\right\rangle
_{\rho}\right],
\end{equation}
where $\Gamma\equiv G_{\rho}-G_{\sigma}$ is symmetric. To evaluate the expectation
$\left\langle \hat{x}^{T}\Gamma\hat{x}\right\rangle _{\rho}$, we can use that $\hat{x}_{k}\hat{x}_{l}=\frac{1}{2}(\left\{  \hat{x}_{l},\hat{x}
_{k}\right\}  -\left[  \hat{x}_{l},\hat{x}_{k}\right]  )=\frac{1}{2}(\left\{
\hat{x}_{l},\hat{x}_{k}\right\}  -i\Omega_{l,k})$ and write $\left\langle
\hat{x}^{T}\Gamma\hat{x}\right\rangle _{\rho}$ as
\begin{equation}
%\left\langle \hat{x}^{T}\Gamma\hat{x}\right\rangle _{\rho}  &  =
\sum
_{k,l}\Gamma_{k,l}\left\langle \hat{x}_{k}\hat{x}_{l}\right\rangle _{\rho
}%\label{eq:rel-ent-derive-1}
%&  =\frac{1}{2}\sum_{k,l}\Gamma_{k,l}\left\langle \left\{  \hat{x}_{l},\hat
%{x}_{k}\right\}  \right\rangle _{\rho}-\frac{i}{2}\operatorname{Tr}
%\{\Gamma\Omega\}\\
  =\sum_{k,l}\Gamma_{k,l}V_{l,k}^{\rho}=\operatorname{Tr}\{\Gamma V^{\rho}\},
\label{eq:rel-ent-derive-last}
\end{equation}
implying that
$
D(\rho\Vert\sigma)=\left[  \ln(Z_{\sigma}/Z_{\rho})-\operatorname{Tr}\{\Gamma
V^{\rho}\}\right]  /2.
$
For states $\rho$ and $\sigma$ that are not zero mean, one can incorporate a
shift into the above calculation to find that
\begin{equation}
D(\rho\Vert\sigma)=\left[  \ln(Z_{\sigma}/Z_{\rho})-\operatorname{Tr}\{\Gamma
V^{\rho}\}+\gamma^{T}G_{\sigma}\gamma\right]  /2, \label{eq:rel-ent-with-mean}
\end{equation}
where $\gamma\equiv\mu^{\rho}-\mu^{\sigma}$. Alternatively, one can write the
formula for relative entropy as
$
D(\rho\Vert\sigma)=[  \ln(Z_{\sigma})+\operatorname{Tr}\{G_{\sigma
}V^{\rho}\}+\gamma^{T}G_{\sigma}\gamma]  /2
-\sum_{i=1}^{n}g(\nu_{i}^{\rho}-1/2),
$
where $\{\nu_{i}^{\rho}\}_{i}$ are the symplectic eigenvalues of $\rho$ and
$g(x)\equiv(x+1)\ln(x+1)-x\ln x$~\footnote{In this way, we see that the expression is
well defined when $\rho$ does not have full support.}.

\textit{Relative entropy variance for Gaussian states}---The following theorem is our main result.

\begin{theorem} \label{th:main}
For Gaussian states $\rho$ and $\sigma$, the relative entropy variance from
\eqref{eq:relative-entropy-variance}\ is given by
\begin{equation}
V(\rho\Vert\sigma)=\frac{\operatorname{Tr}\{(\Gamma V^{\rho})^2\}}{2}
+\frac{\operatorname{Tr}\{(\Gamma\Omega)^2\}}{8}
+\gamma^{T}G_{\sigma}V^{\rho}G_{\sigma}\gamma, \label{eq:rel-ent-var-Gaussian}
\end{equation}
where $\Gamma\equiv G_{\rho}-G_{\sigma}$, $G_{\rho}$ and $G_{\sigma}$ are
defined from \eqref{eq:G_rho}, $\Omega$ is defined in
\eqref{eq:symplectic-form}, $V^{\rho}$ is defined in
\eqref{eq:covariance-matrices}, and $\gamma\equiv\mu^{\rho}-\mu^{\sigma}$.
\end{theorem}

To begin with, let us suppose that the states $\rho$ and
$\sigma$ have zero mean. The calculation then begins with the definition of
the relative entropy variance and proceeds through a few steps:
\begin{align}
V(\rho\Vert\sigma) &
%=\left\langle \left(  \ln\rho-\ln\sigma-D(\rho
%\Vert\sigma)\right)  ^{2}\right\rangle _{\rho}\\
  =\left\langle \left(  -\tfrac{1}{2}\hat{x}^{T}\Gamma\hat{x}+\tfrac{1}
{2}\left\langle \hat{x}^{T}\Gamma\hat{x}\right\rangle _{\rho}\right)
^{2}\right\rangle _{\rho}\\
&  =\tfrac{1}{4}\left[  \left\langle \left(  \hat{x}^{T}\Gamma\hat{x}\right)
^{2}\right\rangle _{\rho}-\left\langle \hat{x}^{T}\Gamma\hat{x}\right\rangle
_{\rho}^{2}\right]  \\
&  =\tfrac{1}{4}\left[  \left\langle \left(  \hat{x}^{T}\Gamma\hat{x}\right)
^{2}\right\rangle _{\rho}-\left[  \operatorname{Tr}\{\Gamma V^{\rho}\}\right]
^{2}\right]  ,\label{eq:rel-ent-var-1}
\end{align}
where the last line follows from
%\eqref{eq:rel-ent-derive-1}--
\eqref{eq:rel-ent-derive-last}. At this point, it
remains to calculate $\langle(\hat{x}^{T}\Gamma\hat{x})^{2}\rangle_{\rho}$,
which we do in Appendix~\ref{sec:details-zero-mean}. To summarize the calculation, one needs to expand
the operator $(\hat{x}^{T}\Gamma\hat{x})^{2}$, leading to an expression of
order four in the quadrature operators. After employing commutators and
anticommutators to bring this operator into Weyl symmetric form
\cite{PhysRevD.2.2161} and at the same time employing symmetries of the
dihedral subgroup of the symmetric group $S_{4}$, we can invoke Isserlis'
theorem \cite{Isserlis18}\ regarding higher moments of Gaussians to evaluate
it. We find that $\frac{1}{4}\left\langle \left(  \hat{x}^{T}\Gamma\hat{x}\right)  ^{2}\right\rangle
_{\rho}$ is equal to
\begin{equation}
\tfrac{1}{4}\left[  \operatorname{Tr}\{\Gamma V^{\rho}\}\right]
^{2}+\tfrac{1}{2}\operatorname{Tr}\{\Gamma V^{\rho}\Gamma V^{\rho
}\}\label{eq:2nd-moment-term}
+\tfrac{1}{8}\operatorname{Tr}\{\Gamma\Omega\Gamma\Omega\},
\end{equation}
which, after combining with \eqref{eq:rel-ent-var-1}, leads to the formula in
\eqref{eq:rel-ent-var-Gaussian} for zero-mean states. Incorporating a shift
then leads to the full formula in \eqref{eq:rel-ent-var-Gaussian}. We provide
full details of the calculation described  above in
Appendix~\ref{sec:details-zero-mean} and generalize it to arbitrary Gaussian
states in Appendix~\ref{sec:non-zero-mean-states}.
Appendix~\ref{sec:well-defined-not-full-support} argues how the formula is well defined even if $\rho$ does not have full support, and Appendix~\ref{sec:standard-form-cov-matrix}\ provides a further simplification
of the formula for two-mode Gaussian states with covariance matrices in
standard form.

\textit{Application to quantum illumination}---In the setting of quantum illumination a transmitter irradiates a target region basked in thermal noise in which a low-reflectivity object may be embedded. Let $\hat{a}_{S}$ denote
the field-mode annihilation operator for the signal mode which is transmitted.
We take the null hypothesis to be that the object is not there, and if this is
the case, the annihilation operator for the return signal is $\hat{a}_{R}
=\hat{a}_{B}$, where $\hat{a}_{B}$ represents an annihiliation operator for a
bath mode in a thermal state $\theta(N_{B})$ of mean photon number $N_{B}>0$.
We take the alternative hypothesis to be that the object is there, and in this
case, $\hat{a}_{R}=\sqrt{\eta}\hat{a}_{S}+\sqrt{1-\eta}\hat{a}_{B}$, where
$\eta\in(0,1)$ is related to the reflectivity of the object and $\hat{a}_{B}$
is now in a thermal state of mean photon number $N_{B}/\left(  1-\eta\right)
$~\footnote{We take this convention, as in \cite{TEGGLMPS08}, because one would not
expect the amount of thermal noise in the return signal to change depending on
whether the object is present.}.

If we prepare the signal mode in the coherent state $|\sqrt{N_{S}}\rangle$\ of
mean photon number $N_{S}>0$, then the null hypothesis state $\rho
_{\text{coh}}$ is a thermal state $\theta(N_{B})$ with mean vector $(0,0)$ and
covariance matrix $\left(  N_{B}+1/2\right)  I_{2}$, and the alternative
hypothesis state $\sigma_{\text{coh}}$ is a displaced thermal state, with mean
vector $(\sqrt{2\eta N_{S}},0)$ and covariance matrix $\left(  N_{B}
+1/2\right)  I_{2}$. It is also easy to check that the $G$ matrix from
\eqref{eq:G_rho} for both of these states is equal to $2\operatorname{arcoth}
(2N_{B}+1)I_{2}$.

Plugging into the formula for relative entropy and relative entropy variance,
we find that these quantities simplify as follows for the coherent-state
transmitter:
\begin{align}
D(\rho_{\text{coh}}\Vert\sigma_{\text{coh}})  &  =\eta N_{S}\ln(1+1/N_{B}),\\
V(\rho_{\text{coh}}\Vert\sigma_{\text{coh}})  &  =\eta N_{S}\left(
2N_{B}+1\right)  \ln^{2}(1+1/N_{B}).
\label{eq:coh-state-rel-ent-var}
\end{align}
In calculating the above, note that the covariance matrices for $\rho
_{\text{coh}}$ and $\sigma_{\text{coh}}$ are the same, so that $\Gamma=0$ in
this case, and we only need to calculate the terms involving $\gamma$ in
\eqref{eq:rel-ent-with-mean} and \eqref{eq:rel-ent-var-Gaussian}. What we see
is that as the signal photon number $N_{S}$\ increases, so does the first
order term $MD(\rho_{\text{coh}}\Vert\sigma_{\text{coh}})$ in the
Type-II\ error probability exponent, indicating a more rapid convergence to
zero. However, the second order term $\sqrt{Mb}\Phi^{-1}(\varepsilon)$ is
actually decreasing for all $\varepsilon\in(0,1/2)$ as $N_{S}$ increases, due
to the fact that $\Phi^{-1}(\varepsilon)<0$ for this range of $\varepsilon$.

Now if the transmitter has a quantum memory available, then it can store an idler
mode entangled with the signal mode and conduct a quantum illumination
strategy. The state we consider is the two-mode squeezed vacuum, with the
reduced state of the signal mode having mean photon number $N_{S}$. This state
has mean vector equal to zero and covariance matrix given by
$
\begin{bmatrix}
\mu & c\\
c & \mu
\end{bmatrix}
\oplus%
\begin{bmatrix}
\mu & -c\\
-c & \mu
\end{bmatrix}
,
$
where $\mu=N_{S}+1/2$ and $c=\sqrt{\mu^{2}-1/4}$. The null hypothesis state
$\rho_{\text{QI}}$ for this setup has mean vector equal to zero and the covariance
matrix
$
\begin{bmatrix}
N_{B}+1/2 & 0\\
0 & \mu
\end{bmatrix}
\oplus%
\begin{bmatrix}
N_{B}+1/2 & 0\\
0 & \mu
\end{bmatrix}
,
$
implying that the return and idler modes are in a product state. The
alternative hypothesis state $\sigma_{\text{QI}}$\ has mean vector equal to
zero and the covariance matrix
$
\begin{bmatrix}
\gamma & \sqrt{\eta}c\\
\sqrt{\eta}c & \mu
\end{bmatrix}
\oplus
\begin{bmatrix}
\gamma & -\sqrt{\eta}c\\
-\sqrt{\eta}c & \mu
\end{bmatrix}
,
$
where $\gamma\equiv\eta N_{S}+N_{B}+1/2$.

\begin{figure*}
\begin{center}
%\subfloat[Parameters: $N_{S}=20$, $\eta=0.01$, $N_{B}=0.01$, and
%$\varepsilon=0.001$.]{
%\begin{overpic}[width=1.0\columnwidth]{example1.pdf}
%  \put(50,1){\footnotesize $M$}
%  \put(3,35){\footnotesize $R$}
%\end{overpic}
%%\includegraphics[width=\columnwidth]{q-illum-plot-2.pdf}
%}
%\subfloat[Parameters: $N_{S}=0.01$, $\eta=0.01$, $N_{B}=20$, and $\varepsilon=0.01$.]{
%\begin{overpic}[width=1.0\columnwidth]{example2.pdf}
%  \put(50,1){\footnotesize $M$}
%  \put(3,35){\footnotesize $R$}
%  \put(38,41){\colorbox{white}{\hspace{3cm}}}
%  \put(38,39){\colorbox{white}{\hspace{3cm}}}
%  \put(38,37){\colorbox{white}{\hspace{3cm}}}
%  \put(38,35){\colorbox{white}{\hspace{3cm}}}
%  \put(38,33){\colorbox{white}{\hspace{3cm}}}
%  \put(38,31){\colorbox{white}{\hspace{3cm}}}
%  \put(38,29){\colorbox{white}{\hspace{4.1cm}}}
%  \put(38,27){\colorbox{white}{\hspace{4.1cm}}}
%  \put(38.8,39.5){\footnotesize coherent}
%  \put(38.8,35.5){\footnotesize coherent ($M \to \infty$)}
%  \put(38.8,31.5){\footnotesize quantum illuminated}
%  \put(38.8,27.5){\footnotesize quantum illuminated ($M \to \infty$)}
%\end{overpic}
\includegraphics[width=\linewidth]{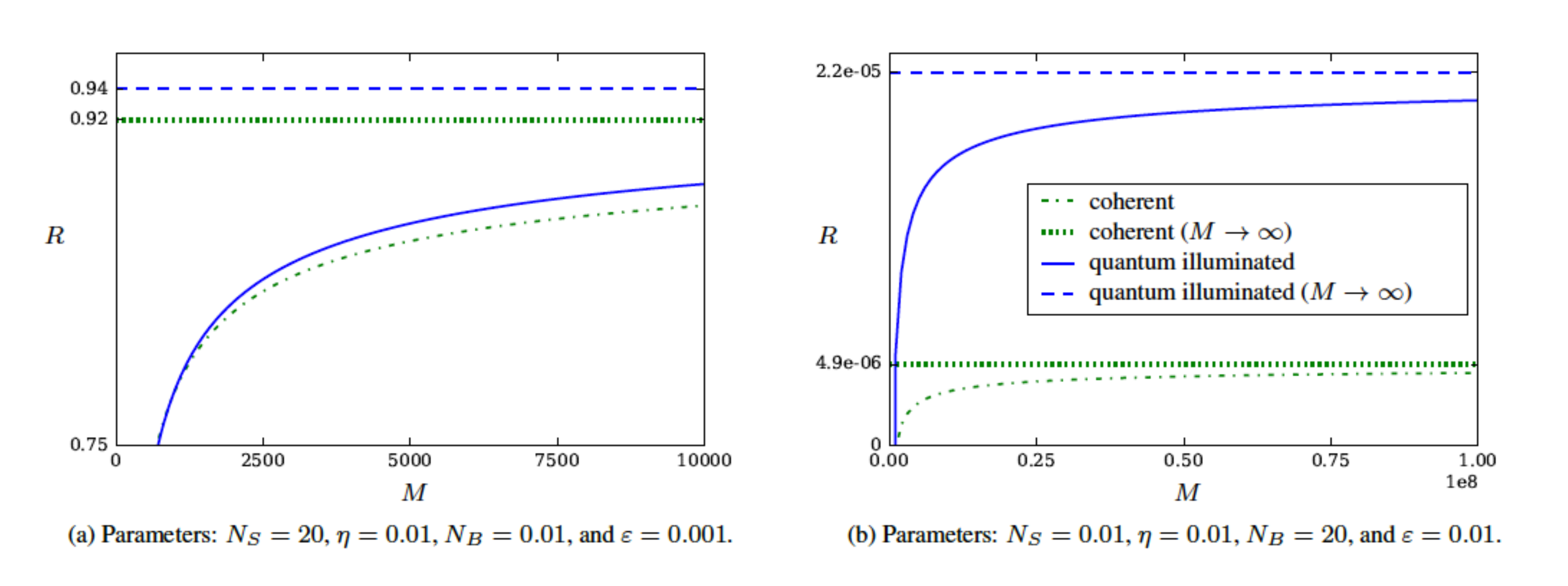}

\end{center}
\caption{Comparison of Type-II\ error probability exponent, $R = -  \ln
\beta /M$, for the quantum illumination transmitter and the
coherent-state transmitter with different parameters.
In both cases, not only does the quantum illumination transmitter achieve
a higher error exponent, but the Gaussian approximation suggests that far
fewer trials are needed to approach this error exponent. The quantum advantage is easier to witness compared to the symmetric-error setting because it occurs for a larger parameter range.
}
\label{fig:comparison}
\end{figure*}

While the expressions for relative entropy and relative entropy variance for the quantum illumination
transmitter are too long to report here, we can evaluate them to first and
second-order in $N_{S}$ (an asymptotic expansion about $N_{S}=\infty$ while keeping $N_B$ fixed), respectively
\begin{align}
D(\rho_{\text{QI}}\Vert\sigma_{\text{QI}})  &  =\frac{\eta N_{S}
}{1-\eta}  \ln
\!\left(1+\frac{1-\eta}{N_{B}}\right)+O(1),\\
V(\rho_{\text{QI}}\Vert\sigma_{\text{QI}})  &  =\left[  \frac{\eta N_{S}
}{1-\eta}  \ln
\!\left(1+\frac{1-\eta}{N_{B}}\right)\right]^2+O(N_{S}).
\end{align}
Alternatively, we can evaluate them to first order in $N_{B}$ (an asymptotic expansion about $N_{B}=\infty$ while keeping $N_S$ fixed):
\begin{align}
D(\rho_{\text{QI}}\Vert\sigma_{\text{QI}})  &  =
\frac{\eta N_{S}(N_{S}+1)
}{N_B}  \ln
\!\left(1+\frac{1}{N_{S}}\right)
+O\!\left(\frac{1}{N_B^2}\right),\\
V(\rho_{\text{QI}}\Vert\sigma_{\text{QI}})  &  =  \frac{\eta N_{S}(N_{S}+1)
(2N_{S}+1)
}{N_B}  \ln^2
\!\left(1+\frac{1}{N_{S}}\right)\nonumber\\
& \qquad\qquad+O\!\left(\frac{1}{N_B^2}\right).
\end{align}
Details about the derivation are in Appendix~\ref{sec:standard-form-cov-matrix}.

%REWRITE: We see that for high $N_{S}$ and any $\eta\in(0,1)$, the quantum
%illumination transmitter substantially outperforms the coherent-state
%transmitter in the first-order term in the Type-II error probability exponent
%(the factor of improvement is $\left[  1-\eta\right]  ^{-1}$). However, for the regime in which $\varepsilon<1/2$ so that
%$\Phi^{-1}(\varepsilon)< 0$,
%the
%second-order term is worse because it is proportional to $N_{S}^{2}$
%rather than $N_{S}$ as in \eqref{eq:coh-state-rel-ent-var}. These findings demonstrate that quantum illumination
%transmitters offer a significant benefit over coherent-state transmitters in
%the setting of many repetitions $M\gg1$, high signal photon number $N_{S}\gg1$, and arbitrary reflectivity $\eta
%\in(0,1)$.
%asymmetric error as considered here

There are several regimes in which the quantum illumination transmitter outperforms the coherent-state transmitter.
We can consider the regime of low background thermal noise, where $N_{S}\gg1$ and $N_{B}\ll1$, and also
%Picking particular
%values $N_{S}=20$, $\eta=0.01$, and $N_{B}=0.01$, a numerical evaluation leads to the values $D(\rho
%_{\text{coh}}\Vert\sigma_{\text{coh}})\approx0.461512$,
%$V(\rho_{\text{coh}}\Vert\sigma_{\text{coh}})\approx2.17253$,
%$D(\rho_{\text{QI}}\Vert\sigma_{\text{QI}})\approx0.941537$, and
%$V(\rho_{\text{QI}}\Vert\sigma_{\text{QI}})\approx5.38013$.
the regime $N_{S}\ll1$ and $N_{B}\gg1$ as considered in
\cite{TEGGLMPS08}.
%Picking particular
%values $N_{S}=0.01$, $\eta=0.01$, and $N_{B}=20$ considered in
%\cite{TEGGLMPS08}, a numerical evaluation leads to the values $D(\rho
%_{\text{coh}}\Vert\sigma_{\text{coh}})\approx2.43951\times10^{-6}$,
%$V(\rho_{\text{coh}}\Vert\sigma_{\text{coh}})\approx4.87998\times10^{-6}$,
%$D(\rho_{\text{QI}}\Vert\sigma_{\text{QI}})\approx7.10588\times10^{-5}$, and
%$V(\rho_{\text{QI}}\Vert\sigma_{\text{QI}})\approx2.06134\times10^{-4}$.
Figures~\ref{fig:comparison}(a) and (b) compare the Type-II\ error probability exponents
of the quantum illumination transmitter and the coherent-state transmitter for
a Type-I error probability $\varepsilon=0.001$ and
$\varepsilon=0.01$, respectively, showing both the first-order
terms and the Gaussian approximations from \eqref{eq:Gaussian-approx}.
%Figure~\ref{fig:comparison}(b) compares the Type~II\ error probability exponents
%of the quantum illumination transmitter and the coherent-state transmitter for
%a Type~I error probability $\varepsilon=0.01$, showing both the first-order
%terms and the Gaussian approximations from \eqref{eq:Gaussian-approx}.
Not only does the quantum illumination transmitter outperform the coherent-state
transmitter in exponent, but the Gaussian
approximation indicates that far fewer trials are required to achieve this
gain. Moreover, when compared to the symmetric-error setting, the quantum advantage is even easier to witness because it occurs for a larger range of parameters.

\textit{Discussion}---We have characterized the Type-II error probability
exponent of hypothesis testing of Gaussian states
in terms of the relative entropy and the relative entropy variance of
two Gaussian states. Our formula for the relative entropy variance should find
applications well beyond the settings considered here, especially to
communication tasks for quantum Gaussian channels. As an application of our
result, we find that not only does a quantum illumination strategy outperform
a coherent-state transmitter with respect to error probability exponent, but
in some cases it requires far fewer trials in order to achieve the optimal
error probability exponent.

\begin{acknowledgments}We are grateful to Nilanjana Datta, Saikat Guha, Stefano Pirandola,
and Kaushik Seshadreesan for discussions
and to Jeffrey H.~Shapiro and Quntao Zhuang for feedback on our manuscript. MT, SL, and MB acknowledge the Hearne Institute for Theoretical Physics at Louisiana State University for hosting them for a research visit during spring of 2016. MB acknowledges funding by the SNSF through a fellowship, funding by the Institute for Quantum Information and Matter (IQIM), an NSF Physics Frontiers Center (NFS Grant PHY-1125565) with support of the Gordon and Betty Moore Foundation (GBMF-12500028), and funding support from the ARO grant for Research on Quantum Algorithms at the IQIM (W911NF-12-1-0521).
SL acknowledges ARO, AFOSR, and IARPA. MT is funded by an ARC Discovery Early Career Researcher Award fellowship (Grant No. DE160100821). MMW acknowledges the NSF
under Award No.~CCF-1350397.
\end{acknowledgments}

\bibliographystyle{unsrt}
\bibliography{Ref}

\begin{thebibliography}{10}

\bibitem{LR08}
Erich~L. Lehmann and Joseph~P. Romano.
\newblock {\em Testing Statistical Hypotheses}.
\newblock Springer Texts in Statistics. Springer, third edition, August 2008.

\bibitem{B74}
Richard Blahut.
\newblock Hypothesis testing and information theory.
\newblock {\em IEEE Transactions on Information Theory}, 20(4):405--417, July
  1974.

\bibitem{B60}
Raghu~R. Bahadur.
\newblock On the asymptotic efficiency of tests and estimates.
\newblock {\em Sankhya: The Indian Journal of Statistics (1933-1960)},
  22(3/4):229--252, 1960.

\bibitem{chernoff1952}
Herman Chernoff.
\newblock A measure of asymptotic efficiency for tests of a hypothesis based on
  the sum of observations.
\newblock {\em The Annals of Mathematical Statistics}, 23(4):493--507, December
  1952.

\bibitem{PhysRevLett.98.160501}
K.~M.~R. Audenaert, J.~Calsamiglia, R.~Mu{\~{n}}oz Tapia, E.~Bagan, Ll.
  Masanes, A.~Acin, and F.~Verstraete.
\newblock Discriminating states: The quantum {Chernoff} bound.
\newblock {\em Physical Review Letters}, 98(16):160501, April 2007.
\newblock arXiv:quant-ph/0610027.

\bibitem{ANSV08}
K.~M.~R. Audenaert, M.~Nussbaum, A.~Szkola, and F.~Verstraete.
\newblock Asymptotic error rates in quantum hypothesis testing.
\newblock {\em Communications in Mathematical Physics}, 279(1):251--283, April
  2008.
\newblock arXiv:0708.4282.

\bibitem{HP91}
Fumio Hiai and D\'enes Petz.
\newblock The proper formula for relative entropy and its asymptotics in
  quantum probability.
\newblock {\em Communications in Mathematical Physics}, 143(1):99--114,
  December 1991.

\bibitem{ON00}
Tomohiro Ogawa and Hiroshi Nagaoka.
\newblock Strong converse and {Stein's} lemma in quantum hypothesis testing.
\newblock {\em IEEE Transactions on Information Theory}, 46(7):2428--2433,
  November 2000.
\newblock arXiv:quant-ph/9906090.

\bibitem{li12}
Ke~Li.
\newblock Second order asymptotics for quantum hypothesis testing.
\newblock {\em Annals of Statistics}, 42(1):171--189, February 2014.
\newblock arXiv:1208.1400.

\bibitem{TH12}
Marco Tomamichel and Masahito Hayashi.
\newblock A hierarchy of information quantities for finite block length
  analysis of quantum tasks.
\newblock {\em IEEE Transactions on Information Theory}, 59(11):7693--7710,
  November 2013.
\newblock arXiv:1208.1478.

\bibitem{CMMAB08}
John Calsamiglia, Ramon Mu\~noz Tapia, Lluis Masanes, Antonio Acin, and Emilio
  Bagan.
\newblock Quantum {Chernoff} bound as a measure of distinguishability between
  density matrices: Application to qubit and {Gaussian} states.
\newblock {\em Physical Review A}, 77(3):032311, March 2008.
\newblock arXiv:0708.2343.

\bibitem{JOPS12}
V.~Jaksic, Y.~Ogata, C.-A. Pillet, and R.~Seiringer.
\newblock Quantum hypothesis testing and non-equilibrium statistical mechanics.
\newblock {\em Reviews in Mathematical Physics}, 24(06):1230002, 2012.
\newblock arXiv:1109.3804.

\bibitem{DPR15}
Nilanjana Datta, Yan Pautrat, and Cambyse Rouz\'{e}.
\newblock Second-order asymptotics for quantum hypothesis testing in settings
  beyond i.i.d. - quantum lattice systems and more.
\newblock {\em Journal of Mathematical Physics}, 57(6):062207, June 2016.
\newblock arXiv:1510.04682.

\bibitem{L08}
Seth Lloyd.
\newblock Enhanced sensitivity of photodetection via quantum illumination.
\newblock {\em Science}, 321(5895):1463--1465, September 2008.
\newblock arXiv:0803.2022.

\bibitem{TEGGLMPS08}
Si-Hui Tan, Baris~I. Erkmen, Vittorio Giovannetti, Saikat Guha, Seth Lloyd,
  Lorenzo Maccone, Stefano Pirandola, and Jeffrey~H. Shapiro.
\newblock Quantum illumination with {Gaussian} states.
\newblock {\em Physical Review Letters}, 101(25):253601, December 2008.
\newblock arXiv:0810.0534.

\bibitem{S09}
Jeffrey~H. Shapiro.
\newblock Defeating passive eavesdropping with quantum illumination.
\newblock {\em Physical Review A}, 80(2):022320, August 2009.
\newblock arXiv:0904.2490.

\bibitem{LRDOBG13}
E.~D. Lopaeva, I.~Ruo~Berchera, I.~P. Degiovanni, S.~Olivares, G.~Brida, and
  M.~Genovese.
\newblock Experimental realization of quantum illumination.
\newblock {\em Physical Review Letters}, 110(15):153603, April 2013.
\newblock arXiv:1303.4304.

\bibitem{ZTZWS13}
Zheshen Zhang, Maria Tengner, Tian Zhong, Franco N.~C. Wong, and Jeffrey~H.
  Shapiro.
\newblock Entanglement's benefit survives an entanglement-breaking channel.
\newblock {\em Physical Review Letters}, 111(1):010501, July 2013.
\newblock arXiv:1303.5343.

\bibitem{ZMWS15}
Zheshen Zhang, Sara Mouradian, Franco N.~C. Wong, and Jeffrey~H. Shapiro.
\newblock Entanglement-enhanced sensing in a lossy and noisy environment.
\newblock {\em Physical Review Letters}, 114(11):110506, March 2015.
\newblock arXiv:1411.5969.

\bibitem{BGWVSP15}
Shabir Barzanjeh, Saikat Guha, Christian Weedbrook, David Vitali, Jeffrey~H.
  Shapiro, and Stefano Pirandola.
\newblock Microwave quantum illumination.
\newblock {\em Physical Review Letters}, 114(8):080503, February 2015.
\newblock arXiv:1410.4008.

\bibitem{WTB16}
Mark~M. Wilde, Marco Tomamichel, and Mario Berta.
\newblock Converse bounds for private communication over quantum channels.
\newblock {\em IEEE Transactions on Information Theory}, 63(3):1792--1817,
  March 2017.
\newblock arXiv:1602.08898.

\bibitem{KW17}
Eneet Kaur and Mark~M. Wilde.
\newblock Upper bounds on secret key agreement over lossy thermal bosonic
  channels.
\newblock June 2017.
\newblock arXiv:1706.04590.

\bibitem{TT13}
Marco Tomamichel and Vincent Y.~F. Tan.
\newblock Second-order asymptotics for the classical capacity of image-additive
  quantum channels.
\newblock {\em Communications in Mathematical Physics}, 338(1):103--137, August
  2015.
\newblock arXiv:1308.6503.

\bibitem{WRG16}
Mark~M. Wilde, Joseph~M. Renes, and Saikat Guha.
\newblock Second-order coding rates for pure-loss bosonic channels.
\newblock {\em Quantum Information Processing}, 15(3):1289--1308, March 2016.
\newblock arXiv:1408.5328.

\bibitem{DTW14}
Nilanjana Datta, Marco Tomamichel, and Mark~M. Wilde.
\newblock On the second-order asymptotics for entanglement-assisted
  communication.
\newblock {\em Quantum Information Processing}, 15(6):2569--2591, June 2016.
\newblock arXiv:1405.1797.

\bibitem{BDL15}
Salman Beigi, Nilanjana Datta, and Felix Leditzky.
\newblock Decoding quantum information via the {Petz} recovery map.
\newblock {\em Journal of Mathematical Physics}, 57(8):082203, August 2016.
\newblock arXiv:1504.04449.

\bibitem{TBR15}
Marco Tomamichel, Mario Berta, and Joseph~M. Renes.
\newblock Quantum coding with finite resources.
\newblock {\em Nature Communications}, 7:11419, May 2016.
\newblock arXiv:1504.04617.

\bibitem{U62}
Hisaharu Umegaki.
\newblock Conditional expectations in an operator algebra {IV} (entropy and
  information).
\newblock {\em Kodai Mathematical Seminar Reports}, 14(2):59--85, 1962.

\bibitem{datta_private}
Nilanjana Datta.
\newblock Private communication.
\newblock 2017.

\bibitem{PhysRevA.71.062320}
Xiao-yu Chen.
\newblock Gaussian relative entropy of entanglement.
\newblock {\em Physical Review A}, 71(6):062320, June 2005.
\newblock arXiv:quant-ph/0402109.

\bibitem{PLOB15}
Stefano Pirandola, Riccardo Laurenza, Carlo Ottaviani, and Leonardo Banchi.
\newblock Fundamental limits of repeaterless quantum communications.
\newblock October 2015.
\newblock arXiv:1510.08863v5.

\bibitem{Arvind1995}
Arvind, B.~Dutta, N.~Mukunda, and R.~Simon.
\newblock The real symplectic groups in quantum mechanics and optics.
\newblock {\em Pramana}, 45(6):471--497, December 1995.
\newblock arXiv:quant-ph/9509002.

\bibitem{adesso14}
Gerardo Adesso, Sammy Ragy, and Antony~R. Lee.
\newblock Continuous variable quantum information: {Gaussian} states and
  beyond.
\newblock {\em Open Systems and Information Dynamics}, 21(01--02):1440001, June
  2014.
\newblock arXiv:1401.4679.

\bibitem{SB14}
Gaetana Spedalieri and Samuel~L. Braunstein.
\newblock Asymmetric quantum hypothesis testing with {Gaussian} states.
\newblock {\em Physical Review A}, 90(5):052307, November 2014.
\newblock arXiv:1407.0884.

\bibitem{PhysRevA.76.062301}
Masahito Hayashi.
\newblock Error exponent in asymmetric quantum hypothesis testing and its
  application to classical-quantum channel coding.
\newblock {\em Physical Review A}, 76(6):062301, December 2007.
\newblock arXiv:quant-ph/0611013.

\bibitem{N06}
Hiroshi Nagaoka.
\newblock The converse part of the theorem for quantum {Hoeffding} bound.
\newblock November 2006.
\newblock arXiv:quant-ph/0611289.

\bibitem{0305-4470-35-50-307}
Masahito Hayashi.
\newblock Optimal sequence of quantum measurements in the sense of {S}tein's
  lemma in quantum hypothesis testing.
\newblock {\em Journal of Physics A: Mathematical and General}, 35(50):10759,
  2002.

\bibitem{W36}
John Williamson.
\newblock On the algebraic problem concerning the normal forms of linear
  dynamical systems.
\newblock {\em American Journal of Mathematics}, 58(1):141--163, January 1936.

\bibitem{K06}
Ole Krueger.
\newblock {\em Quantum Information Theory with {Gaussian} Systems}.
\newblock PhD thesis, Technische Universit\"at Braunschweig, April 2006.
\newblock Available at
  \url{https://publikationsserver.tu-braunschweig.de/receive/dbbs_mods_00020741}.

\bibitem{H10}
Alexander~S. Holevo.
\newblock The entropy gain of infinite-dimensional quantum channels.
\newblock {\em Doklady Mathematics}, 82(2):730--731, October 2010.
\newblock arXiv:1003.5765.

\bibitem{H13book}
Alexander~S. Holevo.
\newblock {\em Quantum systems, channels, information: a mathematical
  introduction}, volume~16.
\newblock Walter de Gruyter, 2012.

\bibitem{Note1}
In this way, we see that the expression is well defined when $\rho $ does not
  have full support.

\bibitem{PhysRevD.2.2161}
Girish~S. Agarwal and Emil Wolf.
\newblock Calculus for functions of noncommuting operators and general
  phase-space methods in quantum mechanics. {I. Mapping} theorems and ordering
  of functions of noncommuting operators.
\newblock {\em Physical Review D}, 2(10):2161--2186, November 1970.

\bibitem{Isserlis18}
Leon Isserlis.
\newblock On a formula for the product-moment coefficient of any order of a
  normal frequency distribution in any number of variables.
\newblock {\em Biometrika}, 12(1/2):134--139, 1918.

\bibitem{Note2}
We take this convention, as in \cite {TEGGLMPS08}, because one would not expect
  the amount of thermal noise in the return signal to change depending on
  whether the object is present.

\bibitem{DGCZ00}
Lu-Ming Duan, G.~Giedke, J.~I. Cirac, and P.~Zoller.
\newblock Inseparability criterion for continuous variable systems.
\newblock {\em Physical Review Letters}, 84(12):2722--2725, March 2000.
\newblock arXiv:quant-ph/9908056.

\bibitem{S00}
R.~Simon.
\newblock {Peres-Horodecki} separability criterion for continuous variable
  systems.
\newblock {\em Physical Review Letters}, 84(12):2726--2729, March 2000.
\newblock arXiv:quant-ph/9909044.

\end{thebibliography}

\pagebreak

\appendix\onecolumngrid

\section{Diagonalizability and symplectic decompositions}

\label{sec:matrix-diagonalizable}Here we show that the matrix $iV^{\rho}
\Omega$ is diagonalizable, which allows for determining the symplectic
eigenvalues and symplectic matrix for any covariance matrix $V^{\rho}$. In
turn, this allows for evaluating the matrix function $\operatorname{arcoth}
(2iV^{\rho}\Omega)$. Consider that
\begin{align}
iV^{\rho}\Omega &  =iS^{\rho}\left(  D^{\rho}\oplus D^{\rho}\right)  \left(
S^{\rho}\right)  ^{T}\Omega=iS^{\rho}\left(  D^{\rho}\oplus D^{\rho}\right)
\Omega\left(  S^{\rho}\right)  ^{-1}\\
&  =S^{\rho}\left(  I_{2}\otimes D^{\rho}\right)  \left(  -\sigma_{Y}\otimes
I_{n}\right)  \left(  S^{\rho}\right)  ^{-1}=S^{\rho}\left(  -\sigma
_{Y}\otimes D^{\rho}\right)  \left(  S^{\rho}\right)  ^{-1},
\end{align}
where in the second equality we used that $S^{T}\Omega S=\Omega$ (implying
$S^{T}\Omega=\Omega S^{-1}$) and in the next that \eqref{eq:symplectic-form}
implies $i\Omega=-\sigma_{Y}\otimes I_{n}$, where $\sigma_{Y}$ is a Pauli
matrix. Since $-\sigma_{Y}$ is diagonalizable as $-\sigma_{Y}=U\left(
-\sigma_{Z}\right)  U^{\dag}$, where
\begin{equation}
U\equiv\frac{1}{\sqrt{2}}
\begin{bmatrix}
1 & 1\\
i & -i
\end{bmatrix}
,
\end{equation}
this implies that we can write
\begin{equation}
iV^{\rho}\Omega=S^{\rho}\left(  U\otimes I_{n}\right)  \left(  -\sigma
_{Z}\otimes D^{\rho}\right)  \left(  U^{\dag}\otimes I_{n}\right)  \left(
S^{\rho}\right)  ^{-1}=S^{\rho}\left(  U\otimes I_{n}\right)  \left(
[-D^{\rho}]\oplus D^{\rho}\right)  \left(  U^{\dag}\otimes I_{n}\right)  \left(
S^{\rho}\right)  ^{-1}.\label{eq:spectral-decompose-iVOm}
\end{equation}

From this last equality, we see that the symplectic decomposition of $V^{\rho
}$ can be computed from the ordinary eigendecomposition of $iV^{\rho}\Omega$.
The symplectic eigenvalues are the entries along the diagonal matrix
$-D^{\rho}\oplus D^{\rho}$ and the symplectic matrix $S^{\rho}$ can be
computed from the matrix $S^{\rho}\left(  U\otimes I_{n}\right)  $ of
eigenvectors of $iV^{\rho}\Omega$ after right-multiplying by $U^{\dag}\otimes
I_{n}$.

At the same time, we see from \eqref{eq:spectral-decompose-iVOm} that the matrix function $\operatorname{arcoth}
(2iV^{\rho}\Omega)$ can be evaluated as
\begin{align}
\operatorname{arcoth}(2iV^{\rho}\Omega)  &  =S^{\rho}\left(  U\otimes
I_{n}\right)  \operatorname{arcoth}\left[  2\left(  [-D^{\rho}]\oplus D^{\rho
}\right)  \right]  \left(  U^{\dag}\otimes I_{n}\right)  \left(  S^{\rho
}\right)  ^{-1}\\
&  =S^{\rho}\left(  U\otimes I_{n}\right)  \left[  \operatorname{arcoth}
(-2D^{\rho})\oplus\operatorname{arcoth}(2D^{\rho})\right]  \left(  U^{\dag
}\otimes I_{n}\right)  \left(  S^{\rho}\right)  ^{-1}\\
&  =S^{\rho}\left(  U\otimes I_{n}\right)  \left[  -\operatorname{arcoth}
(2D^{\rho})\oplus\operatorname{arcoth}(2D^{\rho})\right]  \left(  U^{\dag
}\otimes I_{n}\right)  \left(  S^{\rho}\right)  ^{-1}\\
&  =S^{\rho}\left(  U\otimes I_{n}\right)  \left(  -\sigma_{Z}\otimes
\operatorname{arcoth}\left[  2D^{\rho}\right]  \right)  \left(  U^{\dag
}\otimes I_{n}\right)  \left(  S^{\rho}\right)  ^{-1}\\
&  =S^{\rho}\left(  -\sigma_{Y}\otimes\operatorname{arcoth}\left[  2D^{\rho
}\right]  \right)  \left(  S^{\rho}\right)  ^{-1}\\
&  =S^{\rho}\left(  I_{2}\otimes\operatorname{arcoth}\left[  2D^{\rho}\right]
\right)  \left(  -\sigma_{Y}\otimes I_{n}\right)  \left(  S^{\rho}\right)
^{-1}\\
&  =S^{\rho}\left(  I_{2}\otimes\operatorname{arcoth}\left[  2D^{\rho}\right]
\right)  i\Omega\left(  S^{\rho}\right)  ^{-1},
\end{align}
and so we find that
\begin{align}
G_{\rho}  &  =2i\Omega\operatorname{arcoth}(2iV^{\rho}\Omega)=2i\Omega
S^{\rho}\left(  I_{2}\otimes\operatorname{arcoth}\left[  2D^{\rho}\right]
\right)  i\Omega\left(  S^{\rho}\right)  ^{-1}=-2\Omega S^{\rho}\left(
I_{2}\otimes\operatorname{arcoth}\left[  2D^{\rho}\right]  \right)  \left(
S^{\rho}\right)  ^{T}\Omega\\
&  =-2\Omega S^{\rho}\left[  \operatorname{arcoth}(2D^{\rho})\right]
^{\oplus2}\left(  S^{\rho}\right)  ^{T}\Omega.
\end{align}

\section{Calculation of \eqref{eq:2nd-moment-term} for zero-mean quantum
Gaussian states}

\label{sec:details-zero-mean}This appendix evaluates the expression
$\left\langle \left(  \hat{x}^{T}\Gamma\hat{x}\right)  ^{2}\right\rangle
_{\rho}$ from \eqref{eq:rel-ent-var-1}\ in the main text. Consider that
\begin{align}
\left\langle \left(  \hat{x}^{T}\Gamma\hat{x}\right)  ^{2}\right\rangle
_{\rho} &  =\left\langle \hat{x}^{T}\Gamma\hat{x}\hat{x}^{T}\Gamma\hat
{x}\right\rangle _{\rho}\label{eq:rel-ent-var-4}\\
&  =\sum_{k,l,m,n}\Gamma_{k,l}\Gamma_{m,n}\left\langle \hat{x}_{k}\hat{x}
_{l}\hat{x}_{m}\hat{x}_{n}\right\rangle _{\rho}.
\end{align}
We need to do some manipulations of the expression $\left\langle \hat{x}
_{k}\hat{x}_{l}\hat{x}_{m}\hat{x}_{n}\right\rangle _{\rho}$ in order to get a
sum of all permutations of the operators (this is known as the Weyl symmetric ordering
\cite{PhysRevD.2.2161}---it is necessary for us to use Weyl symmetric ordering because we have defined a Gaussian state to be one
with a Gaussian Wigner representation for which Weyl symmetric ordering is
required, and only under the guarantee of Gaussian statistics can
Isserlis' theorem \cite{Isserlis18} be applied in order to evaluate higher moments). Consider that we can use commutators and
anticommutators to help with this task, so that we can write
\begin{align}
%\left\langle \hat{x}_{k}\hat{x}_{l}\hat{x}_{m}\hat{x}_{n}\right\rangle _{\rho}
%&  =\left\langle \hat{x}_{k}\hat{x}_{m}\hat{x}_{l}\hat{x}_{n}\right\rangle
%_{\rho}+i\Omega_{l,m}\left\langle \hat{x}_{k}\hat{x}_{n}\right\rangle _{\rho
%},\\
\left\langle \hat{x}_{k}\hat{x}_{l}\hat{x}_{m}\hat{x}_{n}\right\rangle _{\rho}
&  =\left\langle \hat{x}_{k}\hat{x}_{m}\hat{x}_{l}\hat{x}_{n}\right\rangle
_{\rho}+i\Omega_{l,m}\left\langle \hat{x}_{k}\hat{x}_{n}\right\rangle \\
&  =\left\langle \hat{x}_{k}\hat{x}_{m}\hat{x}_{n}\hat{x}_{l}\right\rangle
_{\rho}+i\Omega_{l,n}\left\langle \hat{x}_{k}\hat{x}_{m}\right\rangle
+i\Omega_{l,m}\left\langle \hat{x}_{k}\hat{x}_{n}\right\rangle.
\end{align}
Adopting the shorthand $\left\langle klmn\right\rangle \equiv\left\langle
\hat{x}_{k}\hat{x}_{l}\hat{x}_{m}\hat{x}_{n}\right\rangle _{\rho}$ and
$\left\langle kn\right\rangle \equiv\left\langle \hat{x}_{k}\hat{x}
_{n}\right\rangle _{\rho}$, we can then write
\begin{align}
&  \sum_{k,l,m,n}\Gamma_{k,l}\Gamma_{m,n}\left\langle \hat{x}_{k}\hat{x}
_{l}\hat{x}_{m}\hat{x}_{n}\right\rangle _{\rho}\nonumber\\
&  =\frac{1}{3}\sum_{k,l,m,n}\Gamma_{k,l}\Gamma_{m,n}\left[  \left\langle
klmn\right\rangle +\left\langle kmln\right\rangle +\left\langle
kmnl\right\rangle +i\left(  2\Omega_{l,m}\left\langle kn\right\rangle
+\Omega_{l,n}\left\langle km\right\rangle \right)  \right]  \\
&  =\frac{1}{3}\sum_{k,l,m,n}\Gamma_{k,l}\Gamma_{m,n}\left[  \left\langle
klmn\right\rangle +\left\langle kmln\right\rangle +\left\langle
kmnl\right\rangle \right]  +\frac{1}{3}\sum_{k,l,m,n}\Gamma_{k,l}\Gamma
_{m,n}i\left(  2\Omega_{l,m}\left\langle kn\right\rangle +\Omega
_{l,n}\left\langle km\right\rangle \right)  .
\end{align}
We handle these terms one at a time. For the first term consider that the
quantity $\Gamma_{k,l}\Gamma_{m,n}$\ is invariant under the swaps
$k\leftrightarrow l$, $m\leftrightarrow n$, and $\left(  k,l\right)
\leftrightarrow\left(  m,n\right)  $, due to the fact that $\Gamma$ is a
symmetric matrix (note that these swaps realize the dihedral subgroup of
the symmetric group $S_{4}$). Under these various swaps and invariances, the
quantities $\left\langle klmn\right\rangle $, $\left\langle kmln\right\rangle
$, and $\left\langle kmnl\right\rangle $ can realize all 24 permutations of
the letters $klmn$, which implies that
\begin{equation}
\frac{1}{3}\sum_{k,l,m,n}\Gamma_{k,l}\Gamma_{m,n}\left[  \left\langle
klmn\right\rangle +\left\langle kmln\right\rangle +\left\langle
kmnl\right\rangle \right]  =\sum_{k,l,m,n}\Gamma_{k,l}\Gamma_{m,n}\left\langle
\left\{  \hat{x}_{k}\hat{x}_{l}\hat{x}_{m}\hat{x}_{n}\right\}  _{W}
\right\rangle _{\rho},
\end{equation}
where $\left\{  \hat{x}_{k}\hat{x}_{l}\hat{x}_{m}\hat{x}_{n}\right\}  _{W}$
denotes the Weyl symmetric ordering. So now we can employ the fact that $\rho$
is a Gaussian state and a well known formula for the higher moments of
Gaussians (Isserlis' theorem \cite{Isserlis18}) to conclude that
\begin{align}
\sum_{k,l,m,n}\Gamma_{k,l}\Gamma_{m,n}\left\langle \left\{  \hat{x}_{k}\hat
{x}_{l}\hat{x}_{m}\hat{x}_{n}\right\}  _{W}\right\rangle _{\rho} &
=\sum_{k,l,m,n}\Gamma_{k,l}\Gamma_{m,n}\left(  V_{k,l}^{\rho}V_{m,n}^{\rho
}+V_{k,m}^{\rho}V_{l,n}^{\rho}+V_{k,n}^{\rho}V_{l,m}^{\rho}\right)  \\
&  =\left[  \operatorname{Tr}\{\Gamma V^{\rho}\}\right]  ^{2}
+2\operatorname{Tr}\{\Gamma V^{\rho}\Gamma V^{\rho}\}.\label{eq:rel-ent-var-2}
\end{align}

We simplify the other term:%
\begin{equation}
\frac{1}{3}\sum_{k,l,m,n}\Gamma_{k,l}\Gamma_{m,n}i\left(  2\Omega
_{l,m}\left\langle kn\right\rangle +\Omega_{l,n}\left\langle km\right\rangle
\right)  =\frac{2i}{3}\sum_{k,l,m,n}\Gamma_{k,l}\Gamma_{m,n}\Omega
_{l,m}\left\langle kn\right\rangle +\frac{i}{3}\sum_{k,l,m,n}\Gamma
_{k,l}\Gamma_{m,n}\Omega_{l,n}\left\langle km\right\rangle .
\end{equation}
Consider that%
\begin{equation}
\sum_{k,l,m,n}\Gamma_{k,l}\Gamma_{m,n}\Omega_{l,m}\left\langle kn\right\rangle
=\sum_{k,l,m,n}\Gamma_{k,l}\Gamma_{m,n}\Omega_{l,n}\left\langle
km\right\rangle ,
\end{equation}
due to invariance of the quantity $\Gamma_{k,l}\Gamma_{m,n}$ under the swap
$n\leftrightarrow m$, reducing the overall sum to%
\begin{align}
i\sum_{k,l,m,n}\Gamma_{k,l}\Gamma_{m,n}\Omega_{l,m}\left\langle
kn\right\rangle  &  =\frac{i}{2}\sum_{k,l,m,n}\Gamma_{k,l}\Gamma_{m,n}\left(
\Omega_{l,m}\left\langle kn\right\rangle +\Omega_{m,l}\left\langle
nk\right\rangle \right) \\
&  =\frac{i}{2}\sum_{k,l,m,n}\Gamma_{k,l}\Gamma_{m,n}\left(  \Omega
_{l,m}\left\langle kn\right\rangle -\Omega_{l,m}\left\langle nk\right\rangle
\right) \\
&  =-\frac{1}{2}\sum_{k,l,m,n}\Gamma_{k,l}\Gamma_{m,n}\Omega_{l,m}\Omega
_{k,n}\\
&  =\frac{1}{2}\sum_{k,l,m,n}\Gamma_{k,l}\Gamma_{m,n}\Omega_{l,m}\Omega
_{n,k}=\frac{1}{2}\operatorname{Tr}\{\Gamma\Omega\Gamma\Omega\}.
\label{eq:rel-ent-var-3}%
\end{align}
Putting together\ \eqref{eq:rel-ent-var-1}, \eqref{eq:rel-ent-var-4},
\eqref{eq:rel-ent-var-2}, and \eqref{eq:rel-ent-var-3}, we find that%
\begin{align}
V(\rho\Vert\sigma)  &  =\frac{1}{4}\left[  \operatorname{Tr}\{\Gamma V^{\rho
}\}\right]  ^{2}+\frac{1}{2}\operatorname{Tr}\{\Gamma V^{\rho}\Gamma V^{\rho
}\}+\frac{1}{8}\operatorname{Tr}\{\Gamma\Omega\Gamma\Omega\}-\frac{1}%
{4}\left[  \operatorname{Tr}\{\Gamma V^{\rho}\}\right]  ^{2}\\
&  =\frac{1}{2}\operatorname{Tr}\{\Gamma V^{\rho}\Gamma V^{\rho}\}+\frac{1}%
{8}\operatorname{Tr}\{\Gamma\Omega\Gamma\Omega\},
\label{eq:rel-ent-var-zero-mean}%
\end{align}
concluding the proof for zero-mean states.

\section{Relative entropy variance formula for arbitrary Gaussian states}

\label{sec:non-zero-mean-states}

In this appendix we show Theorem~\ref{th:main}, restated here for the reader's convenience.
We compute the relative entropy variance
formula for arbitrary Gaussian states (those that do not necessarily have zero
mean).
\newtheorem*{thmain}{Theorem~\ref{th:main} (restated)}

\begin{thmain}
For Gaussian states $\rho$ and $\sigma$, the relative entropy variance from
\eqref{eq:relative-entropy-variance}\ is given by
\begin{equation}
V(\rho\Vert\sigma)=\frac{\operatorname{Tr}\{(\Gamma V^{\rho})^2\}}{2}
+\frac{\operatorname{Tr}\{(\Gamma\Omega)^2\}}{8}
+\gamma^{T}G_{\sigma}V^{\rho}G_{\sigma}\gamma, %\label{eq:rel-ent-var-Gaussian}
\end{equation}
where $\Gamma\equiv G_{\rho}-G_{\sigma}$, $G_{\rho}$ and $G_{\sigma}$ are
defined from \eqref{eq:G_rho}, $\Omega$ is defined in
\eqref{eq:symplectic-form}, $V^{\rho}$ is defined in
\eqref{eq:covariance-matrices}, and $\gamma\equiv\mu^{\rho}-\mu^{\sigma}$.
\end{thmain}

\begin{proof}
 Here we can see this as a shift of the zero-mean case. Let us define
the centered quadrature vector of operators $\hat{x}_{c}\equiv\hat{x}%
-\mu^{\rho}$ and the difference-of-means vector $\gamma\equiv\mu^{\rho}%
-\mu^{\sigma}$, and then we see that%
\begin{align}
(\hat{x}-\mu^{\sigma})^{T}G_{\sigma}(\hat{x}-\mu^{\sigma}) &  =(\hat{x}%
-\mu^{\rho}+\gamma)^{T}G_{\sigma}(\hat{x}-\mu^{\rho}+\gamma)\\
&  =(\hat{x}_{c}+\gamma)^{T}G_{\sigma}(\hat{x}_{c}+\gamma).
\end{align}
Thus,%
\begin{align}
V(\rho\Vert\sigma) &  =\left\langle \left(  \ln\rho-\ln\sigma-D(\rho
\Vert\sigma)\right)  ^{2}\right\rangle _{\rho}\\
&  =\frac{1}{4}\left\langle \left(  -\hat{x}_{c}^{T}G_{\rho}\hat{x}_{c}%
+(\hat{x}_{c}+\gamma)^{T}G_{\sigma}(\hat{x}_{c}+\gamma)+\left\langle \hat
{x}_{c}^{T}G_{\rho}\hat{x}_{c}-(\hat{x}_{c}+\gamma)^{T}G_{\sigma}(\hat{x}%
_{c}+\gamma)\right\rangle _{\rho}\right)  ^{2}\right\rangle _{\rho}\\
&  =\frac{1}{4}\left\langle \left(  \hat{x}_{c}^{T}G_{\rho}\hat{x}_{c}%
-(\hat{x}_{c}+\gamma)^{T}G_{\sigma}(\hat{x}_{c}+\gamma)+\left\langle -\hat
{x}_{c}^{T}G_{\rho}\hat{x}_{c}+(\hat{x}_{c}+\gamma)^{T}G_{\sigma}(\hat{x}%
_{c}+\gamma)\right\rangle _{\rho}\right)  ^{2}\right\rangle _{\rho}%
\end{align}
Consider that%
\begin{align}
\hat{x}_{c}^{T}G_{\rho}\hat{x}_{c}-(\hat{x}_{c}+\gamma)^{T}G_{\sigma}(\hat
{x}_{c}+\gamma) &  =\hat{x}_{c}^{T}G_{\rho}\hat{x}_{c}-\hat{x}_{c}%
^{T}G_{\sigma}\hat{x}_{c}-\gamma^{T}G_{\sigma}\hat{x}_{c}-\hat{x}_{c}%
^{T}G_{\sigma}\gamma-\gamma^{T}G_{\sigma}\gamma\\
&  =\hat{x}_{c}^{T}\Gamma\hat{x}_{c}-\gamma^{T}G_{\sigma}\hat{x}_{c}-\hat
{x}_{c}^{T}G_{\sigma}\gamma-\gamma^{T}G_{\sigma}\gamma,
\end{align}
which implies that%
\begin{align}
\left\langle \hat{x}_{c}^{T}G_{\rho}\hat{x}_{c}-(\hat{x}_{c}+\gamma
)^{T}G_{\sigma}(\hat{x}_{c}+\gamma)\right\rangle _{\rho} &  =\left\langle
\hat{x}_{c}^{T}\Gamma\hat{x}_{c}\right\rangle _{\rho}-\left\langle \gamma
^{T}G_{\sigma}\hat{x}_{c}\right\rangle _{\rho}-\left\langle \hat{x}_{c}%
^{T}G_{\sigma}\gamma\right\rangle _{\rho}-\left\langle \gamma^{T}G_{\sigma
}\gamma\right\rangle _{\rho}\\
&  =\left\langle \hat{x}_{c}^{T}\Gamma\hat{x}_{c}\right\rangle _{\rho}%
-\gamma^{T}G_{\sigma}\left\langle \hat{x}_{c}\right\rangle _{\rho
}-\left\langle \hat{x}_{c}^{T}\right\rangle _{\rho}G_{\sigma}\gamma-\gamma
^{T}G_{\sigma}\gamma\\
&  =\left\langle \hat{x}_{c}^{T}\Gamma\hat{x}_{c}\right\rangle _{\rho}%
-\gamma^{T}G_{\sigma}\gamma.
\end{align}
Substituting back in above, we find that%
\begin{align}
4V(\rho\Vert\sigma) &  =\left\langle \left(  \hat{x}_{c}^{T}\Gamma\hat{x}%
_{c}-\gamma^{T}G_{\sigma}\hat{x}_{c}-\hat{x}_{c}^{T}G_{\sigma}\gamma
-\gamma^{T}G_{\sigma}\gamma-\left\langle \hat{x}_{c}^{T}\Gamma\hat{x}%
_{c}\right\rangle _{\rho}+\gamma^{T}G_{\sigma}\gamma\right)  ^{2}\right\rangle
_{\rho}\\
&  =\left\langle \left(  \hat{x}_{c}^{T}\Gamma\hat{x}_{c}-\gamma^{T}G_{\sigma
}\hat{x}_{c}-\hat{x}_{c}^{T}G_{\sigma}\gamma-\left\langle \hat{x}_{c}%
^{T}\Gamma\hat{x}_{c}\right\rangle _{\rho}\right)  ^{2}\right\rangle _{\rho}\\
&  =\left\langle \left(  \hat{x}_{c}^{T}\Gamma\hat{x}_{c}-\left\langle \hat
{x}_{c}^{T}\Gamma\hat{x}_{c}\right\rangle _{\rho}-\left[  \gamma^{T}G_{\sigma
}\hat{x}_{c}+\hat{x}_{c}^{T}G_{\sigma}\gamma\right]  \right)  ^{2}%
\right\rangle _{\rho}\\
&  =\left\langle \left(  \hat{x}_{c}^{T}\Gamma\hat{x}_{c}-\left\langle \hat
{x}_{c}^{T}\Gamma\hat{x}_{c}\right\rangle _{\rho}\right)  ^{2}\right\rangle
_{\rho}-\left\langle \left[  \hat{x}_{c}^{T}\Gamma\hat{x}_{c}-\left\langle
\hat{x}_{c}^{T}\Gamma\hat{x}_{c}\right\rangle _{\rho}\right]  \left[
\gamma^{T}G_{\sigma}\hat{x}_{c}+\hat{x}_{c}^{T}G_{\sigma}\gamma\right]
\right\rangle _{\rho}\\
&  \qquad-\left\langle \left[  \gamma^{T}G_{\sigma}\hat{x}_{c}+\hat{x}_{c}%
^{T}G_{\sigma}\gamma\right]  \left[  \hat{x}_{c}^{T}\Gamma\hat{x}%
_{c}-\left\langle \hat{x}_{c}^{T}\Gamma\hat{x}_{c}\right\rangle _{\rho
}\right]  \right\rangle _{\rho}+\left\langle \left(  \gamma^{T}G_{\sigma}%
\hat{x}_{c}+\hat{x}_{c}^{T}G_{\sigma}\gamma\right)  ^{2}\right\rangle _{\rho
}.\label{eq:rel-ent-var-with-means-last-term}%
\end{align}
Consider that the term $\langle(\hat{x}_{c}^{T}\Gamma\hat{x}_{c}-\left\langle
\hat{x}_{c}^{T}\Gamma\hat{x}_{c}\right\rangle _{\rho})^{2}\rangle_{\rho}$ is
the same as what we found for the zero-mean case, and so we already have a
simplified expression for it. It remains to evaluate the latter three terms.
However, the middle two terms are equal to zero. To see this, consider that
any expression involving a product of three quadrature operators, such as
$\left\langle \hat{x}_{c,k}\hat{x}_{c,l}\hat{x}_{c,m}\right\rangle _{\rho}$,
is equal to zero. Consider that%
\begin{align}
\left\langle \hat{x}_{c,k}\hat{x}_{c,l}\hat{x}_{c,m}\right\rangle _{\rho} &
=\frac{1}{2}\left[  \left\langle \hat{x}_{c,k}\hat{x}_{c,l}\hat{x}%
_{c,m}\right\rangle _{\rho}+\left\langle \hat{x}_{c,l}\hat{x}_{c,k}\hat
{x}_{c,m}\right\rangle _{\rho}\right]  +\frac{i}{2}\Omega_{k,l}\left\langle
\hat{x}_{c,m}\right\rangle _{\rho}\\
&  =\frac{1}{2}\left[  \left\langle \hat{x}_{c,k}\hat{x}_{c,l}\hat{x}%
_{c,m}\right\rangle _{\rho}+\left\langle \hat{x}_{c,l}\hat{x}_{c,k}\hat
{x}_{c,m}\right\rangle _{\rho}\right]  ,
\end{align}
where the last line follows because $\left\langle \hat{x}_{c,m}\right\rangle
_{\rho}=0$. By subtracting $\left\langle \hat{x}_{c,k}\hat{x}_{c,l}\hat
{x}_{c,m}\right\rangle _{\rho}/2$, we can conclude that%
\begin{equation}
\left\langle \hat{x}_{c,k}\hat{x}_{c,l}\hat{x}_{c,m}\right\rangle _{\rho
}=\left\langle \hat{x}_{c,l}\hat{x}_{c,k}\hat{x}_{c,m}\right\rangle _{\rho}.
\end{equation}
However, this kind of reasoning could be employed for any swap (and for any
subsequent swap), whence we can conclude that%
\begin{equation}
\left\langle \hat{x}_{c,k}\hat{x}_{c,l}\hat{x}_{c,m}\right\rangle _{\rho
}=\left\langle \left\{  \hat{x}_{c,k}\hat{x}_{c,l}\hat{x}_{c,m}\right\}
_{W}\right\rangle _{\rho}.
\end{equation}
Now we can apply Isserlis' theorem for higher moments of Gaussians to conclude
that $\left\langle \hat{x}_{c,k}\hat{x}_{c,l}\hat{x}_{c,m}\right\rangle
_{\rho}=0$. Thus,%
\begin{equation}
\left\langle \left[  \hat{x}_{c}^{T}\Gamma\hat{x}_{c}-\left\langle \hat{x}%
_{c}^{T}\Gamma\hat{x}_{c}\right\rangle _{\rho}\right]  \left[  \gamma
^{T}G_{\sigma}\hat{x}_{c}+\hat{x}_{c}^{T}G_{\sigma}\gamma\right]
\right\rangle _{\rho}=0,\qquad\left\langle \left[  \gamma^{T}G_{\sigma}\hat
{x}_{c}+\hat{x}_{c}^{T}G_{\sigma}\gamma\right]  \left[  \hat{x}_{c}^{T}%
\Gamma\hat{x}_{c}-\left\langle \hat{x}_{c}^{T}\Gamma\hat{x}_{c}\right\rangle
_{\rho}\right]  \right\rangle _{\rho}=0.\label{eq:odd-term}%
\end{equation}
So it remains to evaluate the last term in
\eqref{eq:rel-ent-var-with-means-last-term}. Consider that $\gamma
^{T}G_{\sigma}\hat{x}_{c}=\hat{x}_{c}^{T}G_{\sigma}\gamma$ because $G_{\sigma
}$ is symmetric. This implies that%
\begin{equation}
\left\langle \left(  \gamma^{T}G_{\sigma}\hat{x}_{c}+\hat{x}_{c}^{T}G_{\sigma
}\gamma\right)  ^{2}\right\rangle _{\rho}=4\left\langle \left(  \gamma
^{T}G_{\sigma}\hat{x}_{c}\right)  ^{2}\right\rangle _{\rho}=4\left\langle
\hat{x}_{c}^{T}G_{\sigma}\gamma\gamma^{T}G_{\sigma}\hat{x}_{c}\right\rangle
_{\rho}=4\left\langle \hat{x}_{c}^{T}rr^{T}\hat{x}_{c}\right\rangle _{\rho},
\end{equation}
where, in the last equality, we have set $r\equiv G_{\sigma}\gamma$.
Continuing,%
\begin{align}
\left\langle \hat{x}_{c}^{T}rr^{T}\hat{x}_{c}\right\rangle _{\rho} &
=\sum_{i,j}\left\langle \hat{x}_{c,i}r_{i}r_{j}\hat{x}_{c,j}\right\rangle
_{\rho}=\sum_{i,j}r_{i}r_{j}\left\langle \hat{x}_{c,i}\hat{x}_{c,j}%
\right\rangle _{\rho}\\
&  =\frac{1}{2}\sum_{i,j}r_{i}r_{j}\left[  \left\langle \left\{  \hat{x}%
_{c,i},\hat{x}_{c,j}\right\}  \right\rangle _{\rho}+i\Omega_{i,j}\right]  \\
&  =\frac{1}{2}\sum_{i,j}r_{i}r_{j}\left\langle \left\{  \hat{x}_{c,i},\hat
{x}_{c,j}\right\}  \right\rangle _{\rho}\\
&  =r^{T}V^{\rho}r=\gamma^{T}G_{\sigma}V^{\rho}G_{\sigma}\gamma
.\label{eq:final-term-rel-ent-var-non-zero-mean}%
\end{align}
Putting together \eqref{eq:rel-ent-var-zero-mean},
\eqref{eq:rel-ent-var-with-means-last-term}, \eqref{eq:odd-term}, and
\eqref{eq:final-term-rel-ent-var-non-zero-mean}, we conclude that%
\begin{equation}
V(\rho\Vert\sigma)=\frac{1}{2}\operatorname{Tr}\{\Gamma V^{\rho}\Gamma
V^{\rho}\}+\frac{1}{8}\operatorname{Tr}\{\Gamma\Omega\Gamma\Omega\}+\gamma^T
G_{\sigma}V^{\rho}G_{\sigma}\gamma.
\end{equation}
\end{proof}

\section{Relative entropy variance formula is well behaved if the second state
is full rank}

\label{sec:well-defined-not-full-support}

In this appendix, we prove that the formula in \eqref{eq:rel-ent-var-Gaussian}
is well defined when $\sigma$ has full support (all symplectic eigenvalues
$>1/2$) and when $\rho$ does not necessarily have full support (some of its
symplectic eigenvalues might be equal to 1/2). We do so by establishing an
alternate formula for the relative entropy variance in terms of the symplectic
eigenvalues and symplectic decompositions of the states $\rho$ and $\sigma$.
Note that we do so only for zero-mean Gaussian states because the extra term
$\gamma^{T}G_{\sigma}V^{\rho}G_{\sigma}\gamma$ in
\eqref{eq:rel-ent-var-Gaussian}\ is finite whenever $\sigma$ has full support.

\begin{proposition}
The relative entropy variance of two zero-mean Gaussian states $\rho$ and
$\sigma$ has the following alternate form:%
\begin{multline}
V(\rho\Vert\sigma)=\operatorname{Tr}\{(A^{\rho})^{2}\left[  (2D^{\rho}%
)^{2}-I\right]  \}-\operatorname{Tr}\{(A^{\rho}\left[  (2D^{\rho}%
)^{2}-I\right]  )^{\oplus2}\widetilde{S}(A^{\sigma})^{\oplus2}\widetilde
{S}^{T}\}\\
+2\operatorname{Tr}\left\{  \left[  \widetilde{S}(A^{\sigma})^{\oplus
2}\widetilde{S}^{T}(D^{\rho})^{\oplus2}\right]  ^{2}\right\}
-\operatorname{Tr}\{(A^{\sigma})^{2}\},
\end{multline}
where $A^{\rho}\equiv\operatorname{arcoth}(2D^{\rho})$, $A^{\sigma}%
\equiv\operatorname{arcoth}(2D^{\sigma})$, and $\widetilde{S}\equiv(S^{\rho
})^{-1}S^{\sigma}$.
\end{proposition}

Before delving into a proof of the above proposition, let us comment on why
the above alternate formula demonstrates that relative entropy variance is
finite when $\rho$ does not necessarily have full support. Consider that the
diagonal matrices $(A^{\rho})^{2}\left[  (2D^{\rho})^{2}-I\right]  $ and
$A^{\rho}\left[  (2D^{\rho})^{2}-I\right]  $ have the following respective
entries for $\lambda\geq1/2$:%
\begin{equation}
f_{1}(\lambda)\equiv\left[  \operatorname{arcoth}(2\lambda)\right]
^{2}\left[  \left(  2\lambda\right)  ^{2}-1\right]  ,\qquad f_{2}%
(\lambda)\equiv\operatorname{arcoth}(2\lambda)\left[  \left(  2\lambda\right)
^{2}-1\right]  ,
\end{equation}
from which we readily see that $\lim_{\lambda\rightarrow1/2}f_{1}%
(\lambda)=\lim_{\lambda\rightarrow1/2}f_{2}(\lambda)=0$ after an application
of L'Hospital's rule. We now proceed with a proof of the above proposition.

\bigskip

\begin{proof}
Our starting point is the formula in \eqref{eq:rel-ent-var-Gaussian} for the
relative entropy variance of two zero-mean Gaussian states:%
\begin{equation}
V(\rho\Vert\sigma)=\frac{1}{2}\operatorname{Tr}\{\Gamma V^{\rho}\Gamma
V^{\rho}\}+\frac{1}{8}\operatorname{Tr}\{\Gamma\Omega\Gamma\Omega\},
\end{equation}
where $\Gamma=G^{\rho}-G^{\sigma}$. Consider that%
\begin{align}
V^{\rho}  &  =S^{\rho}(D^{\rho})^{\oplus2}(S^{\rho})^{T},\\
G^{\rho}  &  =-2\Omega S^{\rho}(A^{\rho})^{\oplus2}(S^{\rho})^{T}\Omega,\\
A^{\rho}  &  =\operatorname{arcoth}(2D^{\rho}).
\end{align}

Expanding the first term, we find that%
\begin{equation}
\operatorname{Tr}\{\Gamma V^{\rho}\Gamma V^{\rho}\}=\operatorname{Tr}%
\{G^{\rho}V^{\rho}G^{\rho}V^{\rho}\}-2\operatorname{Tr}\{G^{\rho}V^{\rho
}G^{\sigma}V^{\rho}\}+\operatorname{Tr}\{G^{\sigma}V^{\rho}G^{\sigma}V^{\rho
}\}.
\end{equation}
We now simplify these one at a time. Consider that%
\begin{align}
\operatorname{Tr}\{G^{\rho}V^{\rho}G^{\rho}V^{\rho}\}  &  =4\operatorname{Tr}%
\{\Omega S^{\rho}(A^{\rho})^{\oplus2}(S^{\rho})^{T}\Omega S^{\rho}(D^{\rho
})^{\oplus2}(S^{\rho})^{T}\Omega S^{\rho}(A^{\rho})^{\oplus2}(S^{\rho}%
)^{T}\Omega S^{\rho}(D^{\rho})^{\oplus2}(S^{\rho})^{T}\}\\
&  =4\operatorname{Tr}\{\Omega(A^{\rho})^{\oplus2}\Omega(D^{\rho})^{\oplus
2}\Omega(A^{\rho})^{\oplus2}\Omega(D^{\rho})^{\oplus2}\}\\
&  =4\operatorname{Tr}\{(A^{\rho})^{\oplus2}(D^{\rho})^{\oplus2}(A^{\rho
})^{\oplus2}(D^{\rho})^{\oplus2}\}\\
&  =4\operatorname{Tr}\{(\left[  A^{\rho}D^{\rho}\right]  ^{2})^{\oplus
2}\}=8\operatorname{Tr}\{\left[  A^{\rho}D^{\rho}\right]  ^{2}%
\}=8\operatorname{Tr}\{(A^{\rho})^{2}(D^{\rho})^{2}\}.
\end{align}
Now, using that $S^{T}\Omega=\Omega S^{-1}$, consider that%
\begin{align}
\operatorname{Tr}\{G^{\rho}V^{\rho}G^{\sigma}V^{\rho}\}  &
=4\operatorname{Tr}\{\Omega S^{\rho}(A^{\rho})^{\oplus2}(S^{\rho})^{T}\Omega
S^{\rho}(D^{\rho})^{\oplus2}(S^{\rho})^{T}\Omega S^{\sigma}(A^{\sigma
})^{\oplus2}(S^{\sigma})^{T}\Omega S^{\rho}(D^{\rho})^{\oplus2}(S^{\rho}%
)^{T}\}\\
&  =4\operatorname{Tr}\{\Omega(A^{\rho})^{\oplus2}\Omega(D^{\rho})^{\oplus
2}\Omega(S^{\rho})^{-1}S^{\sigma}(A^{\sigma})^{\oplus2}(S^{\sigma}%
)^{T}(S^{\rho})^{-T}\Omega(D^{\rho})^{\oplus2}\}\\
&  =4\operatorname{Tr}\{(D^{\rho}A^{\rho}D^{\rho})^{\oplus2}(S^{\rho}%
)^{-1}S^{\sigma}(A^{\sigma})^{\oplus2}(S^{\sigma})^{T}(S^{\rho})^{-T}%
\}=4\operatorname{Tr}\{((D^{\rho})^{2}A^{\rho})^{\oplus2}\widetilde
{S}(A^{\sigma})^{\oplus2}\widetilde{S}^{T}\},
\end{align}
where in the last equality we have used the definition $\widetilde{S}%
\equiv(S^{\rho})^{-1}S^{\sigma}$. Also consider that%
\begin{align}
\operatorname{Tr}\{G^{\sigma}V^{\rho}G^{\sigma}V^{\rho}\}  &
=4\operatorname{Tr}\{\Omega S^{\sigma}(A^{\sigma})^{\oplus2}(S^{\sigma}%
)^{T}\Omega S^{\rho}(D^{\rho})^{\oplus2}(S^{\rho})^{T}\Omega S^{\sigma
}(A^{\sigma})^{\oplus2}(S^{\sigma})^{T}\Omega S^{\rho}(D^{\rho})^{\oplus
2}(S^{\rho})^{T}\}\\
&  =4\operatorname{Tr}\{(S^{\rho})^{-1}S^{\sigma}(A^{\sigma})^{\oplus
2}(S^{\sigma})^{T}(S^{\rho})^{-T}\Omega(D^{\rho})^{\oplus2}\Omega(S^{\rho
})^{-1}S^{\sigma}(A^{\sigma})^{\oplus2}(S^{\sigma})^{T}(S^{\rho})^{-T}%
\Omega(D^{\rho})^{\oplus2}\Omega\}\\
&  =4\operatorname{Tr}\{\widetilde{S}(A^{\sigma})^{\oplus2}\widetilde{S}%
^{T}\Omega(D^{\rho})^{\oplus2}\Omega\widetilde{S}(A^{\sigma})^{\oplus
2}\widetilde{S}^{T}\Omega(D^{\rho})^{\oplus2}\Omega\}\\
&  =4\operatorname{Tr}\{\widetilde{S}(A^{\sigma})^{\oplus2}\widetilde{S}%
^{T}(D^{\rho})^{\oplus2}\widetilde{S}(A^{\sigma})^{\oplus2}\widetilde{S}%
^{T}(D^{\rho})^{\oplus2}\}=4\operatorname{Tr}\left\{  \left[  \widetilde
{S}(A^{\sigma})^{\oplus2}\widetilde{S}^{T}(D^{\rho})^{\oplus2}\right]
^{2}\right\}  .
\end{align}

Now we expand the second term%
\begin{equation}
\operatorname{Tr}\{\Gamma\Omega\Gamma\Omega\}=\operatorname{Tr}\{G^{\rho
}\Omega G^{\rho}\Omega\}-2\operatorname{Tr}\{G^{\rho}\Omega G^{\sigma}%
\Omega\}+\operatorname{Tr}\{G^{\sigma}\Omega G^{\sigma}\Omega\}.
\end{equation}
Consider that%
\begin{align}
\operatorname{Tr}\{G^{\rho}\Omega G^{\rho}\Omega\}  &  =4\operatorname{Tr}%
\{\Omega S^{\rho}(A^{\rho})^{\oplus2}(S^{\rho})^{T}\Omega\Omega\Omega S^{\rho
}(A^{\rho})^{\oplus2}(S^{\rho})^{T}\Omega\Omega\}\\
&  =4\operatorname{Tr}\{\Omega(A^{\rho})^{\oplus2}\Omega(A^{\rho})^{\oplus
2}\}=-4\operatorname{Tr}\{(A^{\rho})^{\oplus2}(A^{\rho})^{\oplus2}\}\\
&  =-8\operatorname{Tr}\{(A^{\rho})^{2}\}.
\end{align}
We also have that%
\begin{align}
\operatorname{Tr}\{G^{\rho}\Omega G^{\sigma}\Omega\}  &  =4\operatorname{Tr}%
\{\Omega S^{\rho}(A^{\rho})^{\oplus2}(S^{\rho})^{T}\Omega\Omega\Omega
S^{\sigma}(A^{\sigma})^{\oplus2}(S^{\sigma})^{T}\Omega\Omega\}\\
&  =4\operatorname{Tr}\{\Omega S^{\rho}(A^{\rho})^{\oplus2}(S^{\rho}%
)^{T}\Omega S^{\sigma}(A^{\sigma})^{\oplus2}(S^{\sigma})^{T}\}\\
&  =4\operatorname{Tr}\{\Omega(A^{\rho})^{\oplus2}\Omega(S^{\rho}%
)^{-1}S^{\sigma}(A^{\sigma})^{\oplus2}(S^{\sigma})^{T}(S^{\rho})^{-T}\}\\
&  =-4\operatorname{Tr}\{(A^{\rho})^{\oplus2}\widetilde{S}(A^{\sigma}%
)^{\oplus2}\widetilde{S}^{T}\}.
\end{align}
By the same calculation above, we see that%
\begin{equation}
\operatorname{Tr}\{G^{\sigma}\Omega G^{\sigma}\Omega\}=-8\operatorname{Tr}%
\{(A^{\sigma})^{2}\}.
\end{equation}

Now we combine all terms together to find that%
\begin{align}
V(\rho\Vert\sigma)  &  =\frac{1}{2}\left[  \operatorname{Tr}\{G^{\rho}V^{\rho
}G^{\rho}V^{\rho}\}-2\operatorname{Tr}\{G^{\rho}V^{\rho}G^{\sigma}V^{\rho
}\}+\operatorname{Tr}\{G^{\sigma}V^{\rho}G^{\sigma}V^{\rho}\}\right]
\nonumber\\
&  \qquad+\frac{1}{8}\left[  \operatorname{Tr}\{G^{\rho}\Omega G^{\rho}%
\Omega\}-2\operatorname{Tr}\{G^{\rho}\Omega G^{\sigma}\Omega
\}+\operatorname{Tr}\{G^{\sigma}\Omega G^{\sigma}\Omega\}\right] \\
&  =4\operatorname{Tr}\{(A^{\rho})^{2}(D^{\rho})^{2}\}-4\operatorname{Tr}%
\{((D^{\rho})^{2}A^{\rho})^{\oplus2}\widetilde{S}(A^{\sigma})^{\oplus
2}\widetilde{S}^{T}\}\nonumber\\
&  \qquad+2\operatorname{Tr}\left\{  \left[  \widetilde{S}(A^{\sigma}%
)^{\oplus2}\widetilde{S}^{T}(D^{\rho})^{\oplus2}\right]  ^{2}\right\}
-\operatorname{Tr}\{(A^{\rho})^{2}\}\nonumber\\
&  \qquad+\operatorname{Tr}\{(A^{\rho})^{\oplus2}\widetilde{S}(A^{\sigma
})^{\oplus2}\widetilde{S}^{T}\}-\operatorname{Tr}\{(A^{\sigma})^{2}\}\\
&  =\operatorname{Tr}\{(A^{\rho})^{2}\left[  (2D^{\rho})^{2}-I\right]
\}-\operatorname{Tr}\{(\left[  (2D^{\rho})^{2}-I\right]  A^{\rho})^{\oplus
2}\widetilde{S}(A^{\sigma})^{\oplus2}\widetilde{S}^{T}\}\nonumber\\
&  \qquad+2\operatorname{Tr}\left\{  \left[  \widetilde{S}(A^{\sigma}%
)^{\oplus2}\widetilde{S}^{T}(D^{\rho})^{\oplus2}\right]  ^{2}\right\}
-\operatorname{Tr}\{(A^{\sigma})^{2}\}.
\end{align}
This concludes the proof.
\end{proof}

\section{Relative entropy variance for two-mode Gaussian states in standard
form}

\label{sec:standard-form-cov-matrix}Two-mode Gaussian states with covariance
matrices in \textquotedblleft standard form\textquotedblright\ have a
covariance matrix as follows \cite{DGCZ00,S00,adesso14}:%
\begin{equation}
V=\left(  I_{2}\oplus\sigma_{Z}\right)  V_{0}^{\oplus2}\left(  I_{2}%
\oplus\sigma_{Z}\right)  =%
\begin{bmatrix}
a & c\\
c & b
\end{bmatrix}
\oplus%
\begin{bmatrix}
a & -c\\
-c & b
\end{bmatrix}
, \label{eq:CM-standard-form}%
\end{equation}
where $\sigma_{Z}$ is the Pauli $Z$ matrix,%
\begin{equation}
V_{0}\equiv%
\begin{bmatrix}
a & c\\
c & b
\end{bmatrix}
, \label{eq:V_0-from-CM}%
\end{equation}
$a,b\geq1/2$, and%
\begin{equation}
c\leq\min\left\{  \sqrt{\left(  a-1/2\right)  \left(  b+1/2\right)  }%
,\sqrt{\left(  a+1/2\right)  \left(  b-1/2\right)  }\right\}  .
\end{equation}
For such states, there are less calculations to perform when calculating the
relative entropy variance due to the extra symmetry that they have. The
symplectic diagonalization of the covariance matrix~$V$ simplifies as well:%
\begin{equation}
V=\left(  I_{2}\oplus\sigma_{Z}\right)  S_{0}^{\oplus2}\left(  I_{2}%
\oplus\sigma_{Z}\right)  D^{\oplus2}\left(  I_{2}\oplus\sigma_{Z}\right)
S_{0}^{\oplus2}\left(  I_{2}\oplus\sigma_{Z}\right)  ,
\end{equation}
where%
\begin{align}
S_{0}  &  \equiv%
\begin{bmatrix}
\omega_{+} & \omega_{-}\\
\omega_{-} & \omega_{+}%
\end{bmatrix}
,\ \ \ \ \ \ \ \ \omega_{\pm}\equiv\sqrt{\frac{a+b\pm\sqrt{y}}{2\sqrt{y}}%
},\ \ \ \ \ \ \ \ D\equiv%
\begin{bmatrix}
\nu_{-} & 0\\
0 & \nu_{+}%
\end{bmatrix}
,\\
\nu_{\pm}  &  \equiv\left[  \sqrt{y}\pm\left(  b-a\right)  \right]
/2,\ \ \ \ \ \ \ \ y\equiv\left(  a+b\right)  ^{2}-4c^{2}.
\end{align}
From this, we can deduce that the $G$ matrix defined in \eqref{eq:G_rho} for
such states simplifies as follows:%
\begin{align}
G  &  =-2\Omega\left(  \left(  I_{2}\oplus\sigma_{Z}\right)  S_{0}^{\oplus
2}\left(  I_{2}\oplus\sigma_{Z}\right)  \right)  \operatorname{arcoth}%
(2D)^{\oplus2}\left(  \left(  I_{2}\oplus\sigma_{Z}\right)  S_{0}^{\oplus
2}\left(  I_{2}\oplus\sigma_{Z}\right)  \right)  \Omega\\
&  =-2\Omega\left(  I_{2}\oplus\sigma_{Z}\right)  S_{0}^{\oplus2}%
\operatorname{arcoth}(2D)^{\oplus2}S_{0}^{\oplus2}\left(  I_{2}\oplus
\sigma_{Z}\right)  \Omega\\
&  =-2\Omega\left(  I_{2}\oplus\sigma_{Z}\right)  \left[  S_{0}%
\operatorname{arcoth}(2D)S_{0}\right]  ^{\oplus2}\left(  I_{2}\oplus\sigma
_{Z}\right)  \Omega.
\end{align}
So then%
\begin{align}
\Gamma &  =G^{\rho}-G^{\sigma}\\
&  =-2\Omega\left(  I_{2}\oplus\sigma_{Z}\right)  \left[  S_{0}^{\rho
}\operatorname{arcoth}(2D^{\rho})S_{0}^{\rho}-S_{0}^{\sigma}%
\operatorname{arcoth}(2D^{\sigma})S_{0}^{\sigma}\right]  ^{\oplus2}\left(
I_{2}\oplus\sigma_{Z}\right)  \Omega\\
&  =-2\Omega\left(  I_{2}\oplus\sigma_{Z}\right)  \left[  \Gamma^{\prime
}\right]  ^{\oplus2}\left(  I_{2}\oplus\sigma_{Z}\right)  \Omega,
\end{align}
where%
\begin{equation}
\Gamma^{\prime}\equiv S_{0}^{\rho}\operatorname{arcoth}(2D^{\rho})S_{0}^{\rho
}-S_{0}^{\sigma}\operatorname{arcoth}(2D^{\sigma})S_{0}^{\sigma}.
\label{eq:Delta-prime}%
\end{equation}

We can now give a simplified formula for the relative entropy variance of
two-mode Gaussian states in standard form:

\begin{lemma}
\label{lem:rel-ent-var-standard-form-simple}The relative entropy variance
$V(\rho\Vert\sigma)$, as defined in \eqref{eq:relative-entropy-variance},
simplifies as follows for two-mode zero-mean Gaussian states $\rho$ and
$\sigma$ with covariance matrices in the standard form
\eqref{eq:CM-standard-form}:%
\begin{equation}
V(\rho\Vert\sigma)=
4\operatorname{Tr}\!\left\{  \left[  \sigma_{Z}\Gamma
^{\prime}\sigma_{Z}V_{0}^{\rho}\right]  ^{2}\right\}  -\operatorname{Tr}%
\{  \sigma_{Z}\Gamma^{\prime}\sigma_{Z}\Gamma^{\prime}\}  ,
\end{equation}
where $V_{0}^{\rho}$ and $\Gamma^{\prime}$ are defined in
\eqref{eq:V_0-from-CM} and \eqref{eq:Delta-prime}, respectively.
\end{lemma}

\begin{proof}
Starting with the formula in \eqref{eq:rel-ent-var-Gaussian}, we find that%
\begin{align}
&  \!\!\!\!\operatorname{Tr}\{\Gamma V^{\rho}\Gamma V^{\rho}\}\\
&  =\operatorname{Tr}\!\left\{  \left[  \left(  -2\Omega\left(  I_{2}%
\oplus\sigma_{Z}\right)  \left[  \Gamma^{\prime}\right]  ^{\oplus2}\left(
I_{2}\oplus\sigma_{Z}\right)  \Omega\right)  \left(  I_{2}\oplus\sigma
_{Z}\right)  \left(  V_{0}^{\rho}\right)  ^{\oplus2}\left(  I_{2}\oplus
\sigma_{Z}\right)  \right]  ^{2}\right\} \\
&  =4\operatorname{Tr}\!\left\{  \left[  \left(  I_{2}\oplus\sigma_{Z}\right)
\Omega\left(  I_{2}\oplus\sigma_{Z}\right)  \left[  \Gamma^{\prime}\right]
^{\oplus2}\left(  I_{2}\oplus\sigma_{Z}\right)  \Omega\left(  I_{2}%
\oplus\sigma_{Z}\right)  \left(  V_{0}^{\rho}\right)  ^{\oplus2}\right]
^{2}\right\} \\
&  =4\operatorname{Tr}\!\left\{  \left[
\begin{bmatrix}
0 & \sigma_{Z}\\
-\sigma_{Z} & 0
\end{bmatrix}
\left[  \Gamma^{\prime}\right]  ^{\oplus2}%
\begin{bmatrix}
0 & \sigma_{Z}\\
-\sigma_{Z} & 0
\end{bmatrix}
\left(  V_{0}^{\rho}\right)  ^{\oplus2}\right]  ^{2}\right\}
=8\operatorname{Tr}\!\left\{  \left[  \sigma_{Z}\Gamma^{\prime}\sigma_{Z}%
V_{0}^{\rho}\right]  ^{2}\right\}  .
\end{align}
We can simplify the other term as well:%
\begin{align}
\operatorname{Tr}\{\Gamma\Omega\Gamma\Omega\}  &  =\operatorname{Tr}\!\left\{
\left[  \left(  -2\Omega\left(  I_{2}\oplus\sigma_{Z}\right)  \left[
\Gamma^{\prime}\right]  ^{\oplus2}\left(  I_{2}\oplus\sigma_{Z}\right)
\Omega\right)  \Omega\right]  ^{2}\right\} \\
&  =4\operatorname{Tr}\!\left\{  \left[  \left(  \Omega\left(  I_{2}\oplus
\sigma_{Z}\right)  \left[  \Gamma^{\prime}\right]  ^{\oplus2}\left(
I_{2}\oplus\sigma_{Z}\right)  \Omega\right)  \Omega\right]  ^{2}\right\} \\
&  =4\operatorname{Tr}\!\left\{  \left[  \Omega\left(  I_{2}\oplus\sigma
_{Z}\right)  \left[  \Gamma^{\prime}\right]  ^{\oplus2}\left(  I_{2}%
\oplus\sigma_{Z}\right)  \right]  ^{2}\right\} \\
&  =4\operatorname{Tr}\!\left\{  \left[  \left(  I_{2}\oplus\sigma_{Z}\right)
\Omega\left(  I_{2}\oplus\sigma_{Z}\right)  \left[  \Gamma^{\prime}\right]
^{\oplus2}\right]  ^{2}\right\} \\
&  =4\operatorname{Tr}\!\left\{  \left[
\begin{bmatrix}
0 & \sigma_{Z}\\
-\sigma_{Z} & 0
\end{bmatrix}
\left[  \Gamma^{\prime}\right]  ^{\oplus2}\right]  ^{2}\right\}
=-8\operatorname{Tr}\{  \sigma_{Z}\Gamma^{\prime}\sigma_{Z}\Gamma
^{\prime}\}  ,
\end{align}
concluding the proof.
\end{proof}

A similar analysis gives the following simplification as well:

\begin{lemma}
The relative entropy $D(\rho\Vert\sigma)$ simplifies as follows for two-mode
zero-mean Gaussian states $\rho$ and $\sigma$ with covariance matrices in the
standard form \eqref{eq:CM-standard-form}:%
\begin{equation}
D(\rho\Vert\sigma)=\left[  \ln Z_{\sigma}+4\operatorname{Tr}\{\sigma_{Z}%
S_{0}^{\sigma}\operatorname{arcoth}(2D^{\sigma})S_{0}^{\sigma}V_{0}^{\rho
}\}\right]  /2-g(\nu_{+}^{\rho}-1/2)-g(\nu_{-}^{\rho}-1/2),
\end{equation}
where $V_{0}^{\rho}$ and $\Gamma^{\prime}$ are defined in
\eqref{eq:V_0-from-CM} and \eqref{eq:Delta-prime}, respectively.
\end{lemma}

\end{document}